%% file: main-smallkey-sort.tex
\documentclass[11pt,letterpaper]{article}
\usepackage[margin=1in]{geometry}

\usepackage{cite}

\usepackage[bookmarks=true,pdfstartview=FitH,colorlinks,linkcolor=darkred,filecolor=darkred,citecolor=darkred,urlcolor=darkred]{hyperref}

\usepackage{colortbl}
\usepackage{mdframed}

\usepackage{paralist}
\usepackage{pdfpages}
\usepackage{enumitem}
\usepackage{float}
\usepackage{booktabs}
\usepackage[override]{cmtt}
\usepackage{color,xcolor}
\usepackage{graphicx,color,eso-pic}
\usepackage{amsmath,amssymb}
\usepackage[
	lambda,
	advantage,
	operators,
	sets,
	adversary,
	landau,
	probability,
	notions,
	logic,
	ff,
	mm,
	primitives,
	events,
	complexity,
	asymptotics,
	keys]
{cryptocode}

\usepackage{tikz}
\usetikzlibrary{fit, shapes, positioning, patterns, arrows,shadows}

\tikzset{
    master/.style={
        execute at end picture={
            \coordinate (lower right) at (current bounding box.south east);
            \coordinate (upper left) at (current bounding box.north west);
        }
    },
    slave/.style={
        execute at end picture={
            \pgfresetboundingbox
            \path (upper left) rectangle (lower right);
        }
    }
}


\renewcommand{\highlightkeyword}[2][\ ]{\ensuremath{\textrm{#2}}#1}

\renewcommand{\poly}{\ensuremath{{{\sf poly}}}\xspace}

\renewcommand{\TC}{\ensuremath{{\bf TC}}}
\newcommand{\LC}{\ensuremath{{\bf LC}}}

\newcommand{\R}{\ensuremath{\mathbb{R}}}

\newcommand{\In}{\textrm{In}}
\newcommand{\eps}{\varepsilon}

\definecolor{darkgreen}{rgb}{0,0.5,0}
\definecolor{lightblue}{RGB}{0,176,240}
\definecolor{darkblue}{RGB}{0,112,192}
\definecolor{lightpurple}{RGB}{124, 66, 168}
\definecolor{grey}{RGB}{139, 137, 137}
\definecolor{maroon}{RGB}{178, 34, 34}
\definecolor{green}{RGB}{34, 139, 34}
\definecolor{types}{RGB}{72, 61, 139}
\definecolor{gold}{rgb}{0.8, 0.33, 0.0}

\definecolor{darkgray}{gray}{0.3}

\newcommand{\skiptext}[1]{}

\usepackage{boxedminipage}
\usepackage{cite}
\usepackage{url}
\usepackage{graphicx}
\usepackage[bf,justification=centering]{caption}
\usepackage[justification=centering]{subcaption}
\usepackage{color}
\usepackage{xspace}
\usepackage{multirow}
\usepackage{amsmath,amsthm,amstext,amssymb,amsfonts,latexsym}
\usepackage{wrapfig}

\definecolor{darkred}{rgb}{0.5, 0, 0}
\definecolor{darkgreen}{rgb}{0, 0.5, 0}
\definecolor{darkblue}{rgb}{0,0,0.5}

\newcommand\markx[2]{}

\renewcommand{\path}{\ensuremath{\mathsf{path}}\xspace}

\newcommand{\N}{\mathbb{N}}

\newcommand{\ignore}[1]{}

\renewcommand{\source}{{s}}

\newcounter{task}

\newtheorem{thm}{Theorem}[section]      

\newtheorem{theorem}[thm]{Theorem}
\newtheorem{conj}[thm]{Conjecture}

\newtheorem{lemma}[thm]{Lemma}
\newtheorem{claim}[thm]{Claim}
\newtheorem{corollary}[thm]{Corollary}
\newtheorem{fact}[thm]{Fact}
\newtheorem{proposition}[thm]{Proposition}

\newtheoremstyle{boxes}
{2pt}
{0pt}
{}
{}
{\bfseries}
{}
{\newline}
{\thmname{#1}\thmnumber{ #2}:  
\thmnote{#3}}

\theoremstyle{boxes}
\newtheorem{myalgo}[thm]{Algorithm}

{
\theoremstyle{definition}
\newtheorem{remark}{Remark}
}

\newcommand{\elaine}[1]{{\footnotesize\color{magenta}[Elaine: #1]}}
\newcommand{\gnote}[1]{{\footnotesize\color{blue}[Gilad: #1]}}
\newcommand{\weikai}[1]{{\footnotesize\color{green}[WK: #1]}}

\renewcommand{\elaine}[1]{}
\renewcommand{\gnote}[1]{}
\renewcommand{\weikai}[1]{}

\newenvironment{MyEnumerate}[1]{\begin{enumerate}\setlength{\itemsep}{0.1cm}
\setlength{\parskip}{-0.05cm} #1}{\end{enumerate}}
\newenvironment{MyItemize}[1]{\begin{itemize}\setlength{\itemsep}{0.1cm}
\setlength{\parskip}{-0.05cm} #1}{\end{itemize}}

\newcounter{cnt:challenge}

\begin{document}
%

\begin{titlepage}
\title{\bf Sorting Short Keys in Circuits of Size $o(n \log n)$}
\author{Gilad Asharov \\ Bar-Ilan University 
\and Wei-Kai Lin \\ Cornell University\and Elaine Shi \\ CMU/Cornell}
\date{{\tt Gilad.Asharov@biu.ac.il, wklin@cs.cornell.edu, runting@gmail.com}}

\maketitle

\begin{abstract}
We consider the classical problem of sorting an input array containing $n$ elements, 
where each element is described with a $k$-bit comparison-key and a $w$-bit payload.
A long-standing open problem  is 
whether there exist $(k + w) \cdot o(n \log n)$-sized boolean circuits for sorting. 
A landmark result in this area
is the work by Ajtai, Koml\'{o}s, and Szemer{\'e}di (STOC'83), where they showed how 
to achieve sorting circuits with $(k + w) \cdot O(n \log n)$ boolean gates. 
The recent work of Farhadi et al.~(STOC'19)
showed that if the famous Li-Li network coding conjecture is true, then sorting 
circuits of size $w \cdot o(n \log n)$ do not exist for general $k$; however, no unconditional lower bound 
is known (in fact proving super-linear circuit lower bounds in general is out
of the reach of existing techniques).

In this paper, we show that 
one can overcome the $n\log n$ barrier when the keys to be sorted are short.
Specifically, we prove that 
there is a circuit
with $(k + w) \cdot O(n k) \cdot \poly(\log^*n - \log^* (w + k))$ boolean gates 
capable of sorting any input array containing $n$ elements, each described with  a $k$-bit key
and a $w$-bit payload.
Therefore, if the keys to be sorted are short, say, $k < o(\log n)$, our result is asymptotically better than
the classical AKS sorting network (ignoring $\poly\log^*$ terms); and 
we also overcome the $n \log n$ barrier in such cases. Such a result might be surprising initially 
because it is long known that comparator-based techniques
must incur $\Omega(n \log n)$
comparator gates even when the keys to be sorted are only $1$-bit
long (e.g., see Knuth's ``Art of Programming'' textbook).
To the best of our knowledge, we are the first to achieve
non-trivial results for sorting circuits using non-comparison-based techniques.
We also show that if the Li-Li network coding conjecture is true, 
our upper bound is optimal, barring $\poly\log^*$ terms, 
for every $k$ as long as $k = O(\log n)$.

\end{abstract}

\end{titlepage}


\input{intro-new}

\input{prelim}

\input{bldgblock}

\input{tightfromloose}

\input{loosefromtight}

\input{mainresult}

\input{selection}

\input{lb}

\input{conclusion}


\section*{Acknowledgments}
Elaine Shi would like to thank Bruce Maggs for explaining the  
AKS sorting network~\cite{aks}, 
for numerous extremely helpful discussions regarding 
the elegant Arora, Leighton, Maggs self-routing permutation network~\cite{alm90}
and Pippenger's 
self-routing super-concentrator~\cite{selfroutingsuperconcentrator},
and for his invaluable feedback on the positioning and 
 presentation of this work.
We also would like to thank Ilan Komargodski 
for his insightful suggestions regarding the positioning of this work. 
We thank the SODA'21 reviewers for their thoughtful feedback.

This work is in part supported by 
an NSF CAREER Award under the award number CNS-1601879, 
a Packard Fellowship, an ONR YIP award, and a DARPA Brandeis award.
 Gilad Asharov is sponsored by the
Israel Science Foundation (grant No. 2439/20), 
and by the BIU Center for Research in Applied
Cryptography and Cyber Security in conjunction with the Israel 
National Cyber Bureau in the
Prime Minister's Office. This project has received funding from the 
European Union's Horizon 2020
research and innovation programme under the Marie Sk lodowska-Curie 
grant agreement No. 891234.

\nocite{optorama}
\bibliographystyle{alpha}
\bibliography{allcited}

\appendix

\input{constants}

\end{document}

%% file: intro-new.tex

\elaine{mention somewhere in intro that aks has been state of the art, and there is zigzag sort but
that's theoretically worse.
some of the techniques come from study of oblivious rams, but it's non-trivial to adapt it to circuit model
and moreover it's interesting that the study of these computational models converge.
reviewer: an algorithmic task that is so basic it's remarkable it has not been completely solved years ago)
} 

\section{Introduction}
\label{sec:intro}
Sorting is one of the most fundamental problems in algorithms design
and complexity theory.
Suppose we want to sort an input array with $n$ elements, each  
described with a $k$-bit comparison key and a $w$-bit payload.
A long-standing open question is whether there exists 
circuits with $(k + w) \cdot o(n \log n)$ boolean gates where each gate
is assumed to have constant fan-in and constant fan-out~\cite{isthereoramlb}.
A landmark result in this space is the work by 
Ajtai, Koml\'{o}s, and Szemer{\'e}di~\cite{aks}, 
where they showed how
to achieve sorting circuits with $O(n \log n)$ comparators --- 
since each comparator can be implemented with $O(k + w)$ boolean gates,
AKS's construction can be implemented as a circuit with $(k + w) \cdot O(n \log n)$ boolean gates.
Although it is well-known that the $n \log n$ barrier is necessary in the 
comparator-based model, a natural question left open by AKS
is whether the $n\log n$ barrier can be overcome in the non-comparison-based model. 
Unfortunately, since AKS, 
this line of work became somewhat stuck in terms of both upper bound and lower bound.

On the upper bound front,
various works have attempted to simplify the AKS construction and/or reduce
its concrete constants~\cite{akspaterson,aksjoel,zigzag} ---  most notably
the recent ZigZag sort of Goodrich (STOC'14)~\cite{zigzag} departed more significantly from the original AKS ---  
unfortunately, none of these works achieved new theoretical improvements and all of them focused on the comparator-based model.
It is worth pointing out that 
in the RAM model, it is long known how to leverage
non-comparison-based techniques to get almost linear-time 
sorting~\cite{sortlinearandersson,Kirkpatricksort,hansort00,hansort01,thorupsort}; but these techniques do not easily translate to the circuit model
since they critically rely on the ability to 
dynamically access memory. Another related work is
that of Lin, Shi, and Xie~\cite{osortsmallkey}, who showed
how to sort short keys on an Oblivious RAM in $o(n \log n)$ time --- 
unfortunately, their algorithm 
is {\it randomized} and there is no obvious way  
to convert it to the circuit model.


\gnote{Break paragraphs. I think we should also move the $\Omega(n \log n)$ comparators that appears after Theorem 1.1, and put after theorem 1.1 just a reminder}
On the lower bound side, there is also little progress:
in fact, proving any unconditional super-linear circuit lower bounds 
is beyond the reach of existing techniques.
The recent work of Fahardi et al.~\cite{sortinglbstoc19} proved a conditional lower bound: 
they showed that 
if the famous Li-Li network coding conjecture is true~\cite{lilinetcoding}, then
the $n\log n$ barrier is inherent for any boolean circuit
that implements sorting. In other words, to construct an $o(n \log n)$ sorting circuit 
would require disproving the Li-Li network coding conjecture.
The Li-Li network coding conjecture, roughly speaking, posits that network coding does not help 
improve the transmission rate relative to multi-commodity flow in undirected graphs, 
in the case of multiple unicast sessions~\cite{lilinetcoding}.
While there is some evidence showing why the conjecture might possibly 
be true~\cite{braverman2016network}, proving or disproving it is beyond the 
reach of existing techniques.
%

Since sorting is a very basic 
primitive in computer science, understanding its circuit complexity 
is an important open direction in theoretical computer science.
In this paper, we make some new 
in this somewhat stagnant front; specifically 
we show
the following novel results.


\paragraph{Main result 1: sorting circuits that overcome the $n\log n$ barrier for small keys.} 
We show 
a result that might be surprising at first sight:
when the keys to be sorted are short, one can in fact overcome the $n\log n$ barrier.
Specifically, we prove the following theorem:

\begin{theorem}
Let $\epsilon > 0$ be an arbitrarily small constant.
$n$ elements
each $w$-bit in length and tagged with a $k$-bit key can be sorted
by a boolean circuit of
size\footnote{\label{fnt:bigwintro} As will be clear later, the case
when $w + k$ is large is easier. Therefore, 
henceforth we often write the expression
$\max\left(1, (\log^* n - \log^*(w + k))^{2 + \epsilon}\right)$ 
simply as $(\log^* n - \log^*(w + k))^{2 + \epsilon}$ 
(i.e., we assume that the expression
is not smaller than $1$).}
$$(k+w)\cdot O(nk) \cdot \max\left(1, (\log^* n - \log^*(w + k))^{2 + \epsilon}\right).$$
\label{thm:introsort}
\end{theorem}

Therefore, for short keys $k = o (\log n)$, our construction 
is asymptotically better than the classical result of AKS~\cite{aks} (ignoring
$\poly\log^*$ terms).
A result of this nature might be initially  
surprising due to a couple natural barriers.
First, the famous 0-1 principle
for sorting (see Knuth's famous textbook~\cite{knuthbook}), 
that any {\it comparator}-based sorting circuit
must consume $\Omega(n \log n)$ comparators even when the keys
to be sorted are only 1 bit long.
Due to the 0-1 principle, our result is necessarily non-comparison-based.
To the best of our knowledge, our work is the first  
to achieve non-trivial results for sorting circuits 
using {\it non-comparison-based} techniques; and thus our techniques depart
significantly from all previous sorting circuit constructions 
which are comparison-based~\cite{aks,zigzag,bitonicsort}.
Note that although non-comparison-based techniques have been exploited
in the RAM model to get almost linear-time sort --- e.g., Radix sort, counting sort, and 
others~\cite{sortlinearandersson,Kirkpatricksort,hansort00,hansort01,thorupsort} ---  
all these algorithms depend heavily on the ability to make input-dependent memory accesses 
and thus they do not have equally efficient counterparts in the circuit model.

Besides the comparison-based barrier, 
another subtlety is that 
for such a result to hold, the sorting circuit necessarily cannot preserve {\it stability}
even for 1-bit keys --- recall that in sorting, stability requires that in the output, elements
with the same key preserve their relative ordering in the input array.
More specifically, 
Lin, Shi, and Xie~\cite{osortsmallkey} showed that
any {\it stable} compaction circuit that treats the payload as {\it indivisible}
must have at least $\Omega(n \log n)$ selector gates --- here 
the indivisibility assumptions requires 
that the circuit only move the payload around through selector gates
but not perform boolean computation on the encoding of the payload.
The indivisibility restriction 
in Lin et al.'s lower bound~\cite{osortsmallkey} can be removed 
if one assumes the Li-Li network coding conjecture~\cite{lilinetcoding} 
--- this is implied by Afshani et al.~\cite{lbmult}'s 
recent work\footnote{Although 
their paper~\cite{lbmult} does not explicitly state the $\Omega(n w \log n)$ 
lower bound result for stable
sorting circuits (even for 1-bit keys), the statement is directly
implied by Theorem 2 of their paper.}.


\paragraph{Main result 2: linear-sized compaction circuit.} 
\ignore{
Compaction (also called tight compaction) is the following problem: given an input array consisting
of $n$ elements each of bit-width $w$ where a subset of the elements are marked as 
distinguished, output
a permutation of the input array where all the distinguished elements 
are moved to the front, and
all non-distinguished elements moved to the end.
}
When the keys to be sorted are each 1 bit long, i.e., $k = 1$, 
the resulting problem is also called
{\it compaction} (or tight compaction)
in the algorithms 
literature~\cite{parallelhash,parallelintegersorting,hanselection,comparetree,beateigen}.
In classical algorithms such as Quick Sort, 
or the linear-time median/selection algorithm by Blum et al.~\cite{blumselect,blummedian}, 
compaction (often called ``partitioning'' in this context) 
is adopted to place elements smaller 
than some pivot to the left, and place elements greater than or equal
to the pivot to the right.

Compaction can be trivially accomplished on a Random Access Machine (RAM)
in linear time 
assuming that each element fits in a memory word: 
one can make a linear scan over the input array
and whenever an element marked with the key $0$ 
is encountered, write it to an output array; and then 
repeat the same for elements marked $1$.
For this reason, compaction was never 
a noteworthy abstraction in the RAM model (although indeed
compaction has been extensively studied 
in other computation models --- see Section~\ref{sec:related} for more details).
It almost seems natural to expect that compaction should also be attainable with 
a linear-sized circuit; but somewhat surprisingly, prior to our work it was not known
how to accomplish this! 
More surprisingly, 
to the best of our knowledge,
to date the best-known {\it explicit} compaction circuit  
is none other but {\it sorting} itself which incurs $\Theta(n w \log n)$ gates!
In hindsight, the reason why such a result 
was not known earlier was exactly due to the couple barriers 
mentioned earlier, namely, a linear-sized compaction circuit cannot
be comparison-based and cannot preserve stability; and thus natural classes
of approaches would fail.

We show how to construct
a (non-comparison-based, non-stable) 
compaction circuit that is $O(n w)\cdot \poly(\log^* n - \log^* w)$ in size --- 
this is both a stepping stone towards achieving Theorem~\ref{thm:introsort} and also a special case 
of Theorem~\ref{thm:introsort} when $k = 1$. 
We summarize the compaction result in the following theorem statement.

\begin{theorem}
For any arbitrarily small constant $\epsilon > 0$, 
there exists a circuit 
that can correctly sort  
any input array containing $n$ elements each of $w$ bits and tagged with a 1-bit key; moreover,
the circuit's size is upper bounded 
by\footnote{Similar to Footnote~\ref{fnt:bigwintro}, we often
write the expression 
$\max\left(1,(\log^*n - \log^*w)^{2 + \epsilon}\right)$
simply as $(\log^*n - \log^*w)^{2 + \epsilon}$, i.e., assuming
that $w$ is not too large such that 
$(\log^*n - \log^*w)^{2 + \epsilon} \geq 1$.
}
\[
O(n w) \cdot \min\left( \max\left(1,(\log^*n - \log^*w)^{2 + \epsilon}\right), \ \ 2^w/w \right)
\] 
As a special case, if $w \geq \log^{(c)} n$
for any arbitrarily large constant $c \geq 1$ or if $w = O(1)$, 
then the circuit size is upper bounded by $O(n w)$.
\label{thm:introcompaction}
\end{theorem}

In the above, $\log^{(c)} m$
means taking iterated logarithm of $m$ for a total of $c$ times; and 
$\log^* m$ outputs the minimum $i$ such that $\log^{(i)} m \leq 1$.
Throughout the paper, $\log$ means $\log_2$ unless noted otherwise.
In essence, we show that indeed a linear-sized compaction circuit exists 
for almost all choices of $w$: when $w$ is at least  
$\log^{(c)} n$
for any constant $c$, 
or when $w$ is a constant.
For the narrow regime 
when $w$ is super-constant but asymptotically smaller than
any $\log^{(c)} n$ for a constant $c$,  
our solution is a $\min((\log^* n -\log^* w)^{2+\epsilon}, 2^w/w)$-factor away from optimal.

\paragraph{Main result 3: linear-sized circuit for selection.}
The selection problem aims to select the $m$ 
smallest element given an input array of $n$ elements
each of bit-width $w$.
Given our linear-sized tight compaction 
circuit, it is not too difficult to combine it with the textbook
median-of-median algorithm~\cite{blumselect,blummedian}: as a corollary, we have 
that median/selection can be computed 
with a boolean circuit 
of $O(nw) \cdot \poly(\log^* n - \log^* w)$ too.
\elaine{maybe elaborate on the problem?}%
To put this result in context, recall that 
Blum et al.~\cite{blumselect,blummedian} showed that median/selection can be accomplished 
in $O(n)$ deterministic time on a RAM --- in fact  
this elegant algorithm has been widely adopted in pedagogy.
It seems natural to expect that computing median/selection should be possible
with a linear-sized circuit; but this was not known until our work 
partly because we do not know how to compute
compaction in a linear-sized circuit before. 
A line of work has indeed 
cared about the circuit complexity 
of selection~\cite{yaoselectnet,sortminmem,jimboselect,pippengerselect}, but
all of the prior work focused on comparator-based  circuits.
It is long known that selection circuits 
in the comparator-based model 
suffer from an $\Omega(n \log n)$ size 
lower bound~\cite{yaoselectnet,sortminmem,jimboselect,pippengerselect};
and thus earlier works in this line~\cite{yaoselectnet,sortminmem,jimboselect} 
focused on tightening the constant in front of the $n \log n$.
A very natural question is whether we can construct asymptotically smaller selection circuits 
using non-comparison-based techniques --- it is almost surprising that
no progress has been made along this front given the fundamental nature of the problem! 
Our result for selection is summarized in the following corollary.

\elaine{TODO: change the technical section, we now select all m smallest}

\begin{corollary}[Linear-sized selection circuit]
For any arbitrarily small constant $\epsilon > 0$, 
there exists a circuit 
that can select 
all $m$ smallest elements 
given any input array containing $n$ elements each of $w$ bits;
and moreover its size is upper bounded 
by 
\[
O(n w) \cdot \min\left( \max\left(1,(\log^*n - \log^*w)^{2 + \epsilon}\right), \ \ 2^w/w \right)
\] 
As a special case, if $w \geq \log^{(c)} n$
for any arbitrarily large constant $c \geq 1$ or if $w = O(1)$, 
then the circuit size is upper bounded by $O(n w)$.
\end{corollary}

We note that the very recent
works of Asharov et al.~\cite{optorama} and a subsequent follow-up~\cite{paracompact}
have shown how to achieve compaction in linear-time on a 
deterministic Oblivious RAM\footnote{
A deterministic Oblivious RAM is one in which the algorithm's memory access 
patterns must be determined only by the size of input (but not the content of input nor randomness).
}.
Their result is of a different nature from 
the above Theorem~\ref{thm:introcompaction} since circuits and Oblivious RAM are incomparable
computation models.
As we explain in more detail in Section~\ref{sec:related}, 
an $O(n)$-time algorithm on an Oblivious RAM  does not directly lead to 
an $O(nw)$-sized circuit; and an $O(nw)$-sized circuit does not 
directly lead to an $O(n)$-time algorithm on an Oblivious RAM.


\paragraph{Main result 4: near optimality of our sorting circuit.} 
Lin, Shi, and Xie~\cite{osortsmallkey} showed that any circuit in the 
indivisible model that sorts $n$ elements 
each with a $k$-bit key
must have at least $\Omega(n k)$ selector gates.
Note that our sorting circuit is indeed in the indivisibility model
and thus we achieve optimality  
for every $k = O(\log n)$ barring $\poly\log^*$ factors.

In our paper, 
we prove a similar lower bound, 
removing the indivisibility 
restriction on the circuit, 
but additionally 
assuming that the famous Li-Li network coding conjecture~\cite{lilinetcoding} is true.
Specifically, we prove the following theorem:

\begin{theorem}
Suppose that the Li-Li network coding conjecture is true (see Conjecture~\ref{conj:main}). Moreover, suppose that 
each element's payload length $w > \log_2 n - k$, and the key length $k \leq \log_2 n$.
Then, any constant fan-in, constant fan-out boolean circuit that can 
sort $n$ elements each with a $k$-bit key 
and a $w$-bit payload must have 
size at least $\Omega(nk \cdot (w - \log_2 n +k))$.
\label{thm:introlb}
\end{theorem}

In comparison, the recent work of Fahardi et al.~\cite{sortinglbstoc19}
implied a special case of the above lower bound  
when there are no constraints on the length of the the keys 
to be sorted\footnote{The statement in their
paper~\cite{sortinglbstoc19} is for the RAM model but their techniques imply
a lower bound in the boolean circuit model for the general case without
any constraints on the key length. }.
Note also that our lower bound 
works only when the payload size is not too small, i.e.,
when $w > \log_2 n - k$.
Since our upper bound works irrespective of $w$ --- in fact our upper bound
works in the indivisible model --- 
Theorem~\ref{thm:introlb}
also shows that for sufficiently large $w$, our upper bound result is asymptotically optimal 
barring $\poly\log^*$ terms.  
For the case of small $w$, 
it remains an interesting open question whether there exists
better upper bounds. However, note that any better upper bound for small $w$
cannot be in the indivisible model
like our algorithm, 
due to 
the $\Omega(n k)$ lower bound 
on the number of selector gates in the indivisible model~\cite{osortsmallkey}.



\section{Technical Highlight}
\label{sec:highlight}

In this section, we provide an informal technical roadmap of our main ideas.
To obtain the sorting result stated in Theorem~\ref{thm:introsort}, 
our blueprint is the following:
\begin{MyEnumerate}
\item 
First, 
we will construct a linear-sized compaction 
circuit --- here when we say linear-sized, we omit the $\poly\log^*$ term for convenience.
Recall that compaction is the problem of sorting elements with 1-bit keys. 
\item 
Once we know how to 
solve the 1-bit special case with a linear-sized circuit,
 we then build on top this idea
to construct a linear-sized circuit for computing selection (i.e., selecting
the $m$-th smallest element).

\item 
Finally, using linear-size compaction and selection as building blocks,  
we show how to construct a circuit that sorts $k$-bit keys 
through a clever 2-parameter recursion.
\end{MyEnumerate}

Below, we begin by describing how to accomplish linear-sized compaction.
We stress that even for this 1-bit special case, it is important that the scheme
be {\it non-comparison-based}, since otherwise due to the 
0-1 principle (see Knuth's ``Art of Programming'' textbook~\cite{knuthbook}), 
the circuit must incur at least $\Omega(n \log n)$ comparators.
In our presentation below, we will point out where our scheme
relies on non-comparison-based techniques, since this is why our techniques
fundamentally depart from previous sorting networks (all of which, to the best of our knowledge,
are comparator-based).

\subsection{Warmup: From Pippenger's Super-Concentrator
 to $O(n \log n + nw)$-Sized Compaction Circuit}
\label{sec:firstattempt}

Pippenger~\cite{selfroutingsuperconcentrator} describes an $O(n)$-sized super-concentrator 
construction with $n$ sources and $n$ destinations, such that 
given any subset of $\source$ sources and  
any subset of $\source$ destinations, a set of vertice-disjoint paths 
exist between the sources and destinations. 
If one thinks of the $\source$ sources as clients, 
and the $\source$ destinations as servers, a super-concentrator
allows each client to receive service from a server 
over {\it congestion-free} paths. Note that all servers 
are assumed to 
provide the same type of service, and thus a client can be matched with any server.
Not only does Pippenger show the existence of vertice-disjoint 
paths between any $\source$ sources and any $\source$ destinations, importantly, 
he describes an explicit, efficient algorithm for 
finding a set of such paths when the $\source$ sources and $\source$ destinations 
have been specified --- henceforth we call this the {\it route-finding} algorithm.

To use such a self-routing super-concentrator
to solve the compaction problem, 
one may think of the $\source$ sources as the $\source$ input elements
marked as distinguished, and the $\source$ destinations as the first $\source$  
positions of the output array, i.e., where the distinguished elements want to go.
Now, a trivial observation 
is that in Pippenger's construction,
the route-finding algorithm can be implemented as 
an $O(n)$-time RAM algorithm.  With some extra work, 
one can show
that in fact, the same route finding algorithm 
can be implemented as an $O(n \log n)$-sized 
circuit: the detailed circuit construction 
involves a few non-trivial technicalities but
since this is arguably not the most exciting part of our techniques, 
we defer the details to 
Sections~\ref{sec:lcprelim}, \ref{sec:lswprelim}, and \ref{sec:tightfromloose}.


Once the routes have been found, the $\source$ distinguished elements 
can be routed to their respective destinations 
with an $O(n w)$-sized circuit.
This results in an $O(n\log n + nw)$-sized 
compaction circuit\footnote{We will actually need something
slightly stronger than merely converting
Pippenger's self-routing superconcentrator 
an $O(n\log n + nw)$-sized 
compaction circuit; see Remark~\ref{rmk:piptechnicality}.}. 

\paragraph{The usage of non-comparison techniques.}
We stress that even this warmup scheme must necessarily be non-comparison-based
since for $w \geq \log n$, the circuit size is upper bounded by $O(nw)$.
Thus we now reflect on how non-comparison-based techniques helped us
achieve this result which otherwise would have been impossible.
The warmup algorithm has a {\it metadata phase} 
(i.e., the route-finding phase) 
which involves
looking at only the distinguished/non-distinguished indicator of each element but not the elements themselves, 
and a {\it routing} phase that actually moves 
the elements around according to the plan computed by the metadata phase.
The metadata phase relies on counting and hence makes the algorithm inherently non-comparison-based; and 
moreover,  
the cost of the metadata phase 
does not depend on $w$, i.e., the bit-width 
of the elements. 
The routing phase, on the other hand, moves $O(n)$ elements around  
since the super-concentrator 
is linear in size (and moving each element around incurs $O(w)$ number of gates 
when implemented in circuit).
This key insight will continue to help us throughout the remainder of the paper.

\subsection{A Slightly Deeper Dive into Pippenger's Construction}
\label{sec:pippenger}

Our next goal is to improve the warmup scheme 
to $O(n w) \cdot \poly(\log^* n - \log^* w)$.
Before explaining our approach in the subsequent Section~\ref{sec:roadmapcompact}, 
it helps to dissect Pippenger's construction further.
For ease of understanding, we focus on achieving tight compaction, 
i.e., routing any $\source$ positions in the input to the first $\source$ positions in the output ---
but keep in mind that both Pippenger's and our 
construction can be generalized to 
routing any $\source$ positions in the input to any $\source$ positions in the output.
Furthermore we explain his algorithm with 
a circuit computation model (although their original description is in a distributed
automata model). 

\paragraph{Loose compaction.}
To achieve tight compaction,  
Pippenger first constructs a gadget capable of a relaxed form of compaction
henceforth called {\it loose compaction}.
In loose compaction, we have an input array of length $n$ where at most $n/128$
elements are {\it real} and the rest are {\it dummy}, 
and we would like to 
compress this array to half of the original size, while preserving
the multiset of real elements. 
Specifically, the loose compaction leverages a bipartite expander graph 
with $O(n)$ vertices 
where there are twice as many vertices on the left than on the right. \elaine{use this version?}%
Each vertex has $O(1)$ edges. 
A route-finding
algorithm can be run to find a way to route the real elements on 
the left to the right without causing congestion, compressing
the array by a half in the process.
With a little more work, we show that it is possible to implement their route-finding
algorithm for loose compaction 
with an $O(n \log n)$-sized circuit (to be spelled out in Section~\ref{sec:lcprelim}). 
Once the routes have been found, the actual routing of the elements
from the left to the right 
can be accomplished with an $O(n w)$-sized circuit since the expander
graph has only $O(n)$ edges in total.
Thus,  the entire loose compaction circuit consumes $O(nw + n\log n)$ 
gates.

\paragraph{Upgrading from loose to tight compaction.}
Pippenger then goes to construct a tight compaction algorithm 
given a loose compaction gadget.
In this paper, we will view this part of the algorithm 
as a way to upgrade a loose compactor to a tight one.
Given a loose compactor, to upgrade it to a tight compactor
requires additional use of certain bipartite
expander graphs 
with appropriate vertex expansion. 
Pippenger's
original description was for a Distributed Finite
Automata model, but later in Section~\ref{sec:tightfromloose},
we will describe how to accomplish this loose-to-tight upgrade
in a circuit model of computation while preserving efficiency in some technical sense. 
This requires resolving additional technicalities which we briefly mention below.

\begin{remark}
One technicality is the following: given a circuit that accomplishes
loose compaction with ${\sf LC}(n, w)$
gates, it is quite easy to implement
Pippenger's ideas in a circuit model
and upgrade it to a tight compaction circuit 
consuming $O({\sf LC}(n, w) + nw + n\log n)$ gates.
In fact, this would be sufficient to give us the warmup 
result, that is, an $O(nw + n\log n)$-sized circuit for tight compaction.
This na\"ive approach, 
however, turns out not to be enough
for the bootstrapping steps needed later in Section~\ref{sec:roadmapcompact} to improve
the circuit size.
It will become clear soon that for the bootstrapping steps, 
we need that the loose-to-tight upgrade
be accomplished with a circuit of 
$O({\sf LC}(n, w) + nw)$ gates: showing that this 
is possible is subtle but can be accomplished through standard techniques
--- we thus defer the details to Section~\ref{sec:tightfromloose}. 
\label{rmk:piptechnicality}
\end{remark}


\ignore{
At a very high level, 
the upgrade works as follows --- below we transcribe Pippenger's original ideas
into a form that will be amenable for the circuit model of computation ({\it c.f.} Pippenger's
original description was for a Distributed Finite 
Automata model):
\begin{itemize}
\item 
Suppose there are $m$ 0-elments, and $n-m$ 1-elements in the input array.
Any input 0-element that appears in the latter $n-m$ positions or 
any 1-element that appears in the first $m$ positions are said to be {\it misplaced}. 
One may now think of the task of tight compaction as 
swapping the misplaced 0-elements with the misplaced 1-elements, such that all 
elements fall into place.

First, rely on a bipartite 
expander graph to implement a gadget called a ``loose swapper''.
The loose swapper 
can swap 0-elements and 1-elements in the input array that are misplaced, 
such that only $1\%$ of the elements remain misplaced.
In other words, a loose swapper can make sure that 
99\% of the input elements are 
in the right place, and only 1\% remains to be swapped\footnote{In our formal description
later, we will use different constants.}.

\item 
At this moment, we rely on the loose compactor 
to compress the array's size by a half, such that  
elements that are misplaced are guaranteed not to be lost 
\end{itemize}
}

\subsection{Linear-Sized Compaction Circuit: Our Approach in a Nutshell}
\label{sec:roadmapcompact}
For simplicity, in the main body of the paper, we first present a construction
with a slightly looser $\poly\log^*$, i.e., not caring about the constant-degree of 
the $\poly(\cdot)$ function.
We shall tighten the constant-degree of the $\poly$ to $2+\epsilon$ in Appendix~\ref{sec:tighten}.

Let $\poly(\cdot)$
denote an appropriate polynomial function;
let $\alpha(n) \leq \log^* n$ be any 
integer function that is upper-bounded by $\log^* n$ everywhere. 
Our construction (implicitly) builds  
an $\big(n 
\cdot \poly(\alpha)\big)$-sized super-concentrator
with a route-finding algorithm that can be implemented
as an $\big(n \cdot \log^{(\alpha)} n \cdot \poly(\alpha)\big)$-sized circuit.
Specifically, if we choose $\alpha(n) = \log^* n - \log^* w$, we have
an $\big(n \cdot \poly(\log^* n - \log^* w)\big)$-sized super-concentrator
whose routing-finding algorithm can be computed by  
an $\big(n w \cdot \poly(\log^* n - \log^* w)\big)$-size circuit.
In comparison with Pippenger's super-concentrator, we somewhat significantly 
reduce the size of the circuit implementing the route-finding algorithm, 
at the price of a slightly super-linear super-concentrator.

\ignore{
The main cost to reduce is that of the route-finding algorithm.
A key insight is that the route-finding algorithm operates
only on metadata and not the input elements themselves.
}
The key insight is a 
new technique that involves repeated bootstrapping and boosting:
we use loose compaction to bootstrap tight compaction 
with a constant blowup in circuit size;
then in turn, we use tight compaction 
to construct an asymptotically better loose compaction; and 
then in turn, we use the better loose compaction to bootstrap a better  
tight compaction, and so on.
See Figure~\ref{fig:demonstration} for a demonstration. 
More specifically,

\elaine{TODO: MOVE figure to later sections}
\input{figure}

\begin{MyItemize}
\item {\it Tight compaction from loose compaction.}
Given a function $1 \leq f(n) \leq \log n$,
if loose compaction
can be accomplished with a $C_{\rm lc} \cdot (n f(n) + nw)$-sized circuit,
then we can accomplish tight compaction 
with a $C \cdot C_{\rm lc} \cdot (n f(n) + nw)$-sized circuit 
where $C$ and $C_{\rm lc}$ are suitable constants.
This uses 
the same techniques as in Pippenger's work but now applying them to the circuit model;

\item {\it Loose compaction from tight compaction.}
Let $f(n)$ be a function satisfying $1 \leq f(n) \leq \log n$,
and given a tight compaction circuit of $C_{\rm tc} \cdot (n f(n) + nw)$ size, 
we can construct a $C \cdot C_{\rm tc} \cdot (n f(f(n)) + nw)$-sized 
loose compaction circuit as follows where $C$ and $C_{\rm tc}$ are suitable constants:
\begin{enumerate}
\item 
Divide the input into chunks of $f(n)$ elements; for simplicity
we shall assume that $n$ is divisible by $f(n)$ here and the indivisible case
is handled later in our technical sections.
Mark chunks that have more than 
$f(n)/32$ real elements as {\it dense} and those with at most $f(n)/32$ 
real elements as {\it sparse}.

Use tight compaction to move all the dense chunks to the 
front and all the sparse chunks to the end.
Note that this step can be implemented as a $C_{\rm tc} \cdot(n + nw)$-sized circuit.

\item 
It is easy to show that there are at least $\frac{3n}{4f(n)}$ sparse chunks.
Now, for each of the trailing $\frac{3n}{4f(n)}$ sparse chunks,
use a tight compaction circuit
to compress away $31 f(n)/32$ dummy elements.
This step requires a $C_{\rm tc}\cdot (f(n) f(f(n)) + f(n) w)$-sized circuit 
per chunk and thus over all chunks, the total circuit size
is $\left(n/f(n)\right) 
\cdot  C_{\rm tc}\cdot (f(n) f(f(n)) + f(n) w) = C_{\rm tc}\cdot  (n f(f(n)) + nw)$.
\item 
The output is simply a concatenation of the following:
{\it i)} the first $\frac{n}{4f(n)}$ chunks; and {\it ii)} 
the remaining $f(n)/32$ elements for each of the remaining $\frac{3n}{4f(n)}$ chunks.  
It is not hard to show that the resulting array 
is less than half of the length of the original and no real element
is lost during this process.
\end{enumerate}
\end{MyItemize}

Now, as a starting point, we have a loose compaction 
circuit denoted $\LC_0$ whose size is $C_0 \cdot (n \log n + nw)$
for loosely compacting an input array of $n$ elements each of $w$ bits, where $C_0$ is
an appropriate constant --- as mentioned in Section~\ref{sec:pippenger},
we can construct such an $\LC_0$ using Pippenger's loose compactor construction but
implementing his construction in a $O(n \log n + nw)$ sized-circuit involves some 
non-trivial technicalities which will be explained in Section~\ref{sec:gadgets}.
With $\LC_0$, we can construct a 
tight compaction circuit denoted $\TC_1$
of size at most $C \cdot C_0 \cdot (n \log n + nw)$.
Given $\TC_1$, we can in turn construct a loose compaction circuit 
denoted $\LC_1$ whose size is upper bounded by 
of size at most $C^2 \cdot C_0 \cdot (n \log \log n + nw)$.
Given  $\LC_1$, we can in turn construct a tight compaction
circuit denoted $\TC_2$ of size at most  
$C^3 \cdot C_0 \cdot (n \log  \log n + nw)$.
Given $\TC_2$, we can in turn construct a loose compaction
circuit denoted $\LC_2$ of size at most 
$C^4 \cdot C_0 \cdot (n \log^{(4)} n + nw)$.

Let $\alpha \leq \log^* n$, 
and let $2^{d-1} = \alpha$. 
After $d$
iterations of the above recursion, we 
will obtain a tight compaction circuit $\TC_{d}$
whose size is upper bounded by 
$C^{2d-1} \cdot C_0 \cdot n \cdot (\log^{(2^{d-1})} n + w) 
= O(1) \cdot \poly(\alpha) \cdot n \cdot (\log^{(\alpha)} n + w)$.
Setting $\alpha = \log^*n - \log^* w$ would give us $\log^{(\alpha)}n = O(w)$ and thus we get 
a tight compaction
circuit of size 
$O(nw) \cdot \poly(\log^* n - \log^* w)$.

\ignore{
\begin{itemize}[leftmargin=5mm]
\item 
First, use an $O(n)$-sized expander 
graph to swap misplaced elements in the input, such that 
the remaining fraction of misplaced elements is bounded
by $1/128$.  
More specifically, 
if in the output, a position $i$ ought to contain a distinguished element
but in the input the position $i$ contains a non-distinguished element instead, then 
this element is colored ${\tt blue}$; 
contrarily, 
if in the output a position $i$ ought to contain a non-distinguished element
but in the input the position contains a distinguished element, then the
element is colored ${\tt red}$.
The total number of ${\tt blue}$ and ${\tt red}$ elements must be equal, and in this step
we use an expander graph to swap ${\tt blue}$ elements with ${\tt red}$ ones
such that the remaining density of  
misplaced elements is at most $1/128$. 

\item 
\end{itemize}
}



\ignore{
Furthermore, he describes an algorithm for finding such vertice-joint paths
when the $k$ sources have been identified.
With some additional work, one can show that Pippenger's path-finding
algorithm can be implemented
as an $O(n\log n)$-sized circuit that is independent  
of the word-length $w$; and moreover 
once a set of paths have been identified, the routing of the elements
to the respective destinations can be accomplished with an $O(n w)$-sized circuit.
}

\subsection{Sorting Circuits for Small Keys}

Now we have solved the 1-bit special case of Theorem~\ref{thm:introsort}, we would like
to generalize the scheme to sort small keys. 
If the 1-bit sorter were stable, then we can rely on radix sort  
to sort the elements one bit at a time. 
Unfortunately, as mentioned, our 1-bit sorter is not and cannot be stable.
Instead, 
we rely on a non-trivial two-parameter recursion 
that is inspired by Lin et al.~\cite{osortsmallkey} --- in comparison, their
approach relies on randomness; but we show how to  
adapt their ideas to a deterministic approach.

\paragraph{Linear-sized selection circuit.}
To describe our sorting circuit result, we will need yet another building block:
a linear-sized selection circuit (ignoring $\poly\log^*$ terms).
It turns out that 
given a linear-sized compaction circuit, it is not too difficult
to construct a linear-sized selection circuit based on Blum et al.~\cite{blumselect}'s idea.
To select the $m$-th smallest element,
Blum's algorithm first divides the elements into $n/5$ groups each containing 5 elements, 
computes
the median of each group, and then recursively computes the median of the medians
henceforth called a {\it pivot}.
At this moment,
perform a {\it partitioning}:
move all elements smaller than or equal to the pivot
to the left and all other elements to the right and let ${\bf Y}$ be the resulting
array.  Now, depending on the choice of $\source$, recurse
on either the first $7n/10$ elements of
${\bf Y}$ or the last $7n/10$ elements of ${\bf Y}$.

Blum et al's selection algorithm runs in linear time on a RAM.  
How can we implement it in the circuit model with 
linear number of gates?
The crux is solving the {\it partitioning} step 
with a linear-sized circuit since it is not too hard to check 
that all other steps are easy to implement in a circuit while preserving efficiency.
Clearly, the partitioning step can be solved with a compaction circuit: 
we can label each element with a 0-key if it is smaller than or equal to the pivot
and a 1-key otherwise.
Using our linear-sized compaction circuit to realize
partitioning in Blum et al.'s algorithm, we can then get a linear-sized
selection circuit.

\paragraph{Our sorting circuit.}
We are now ready to present our sorting circuit construction.
To upgrade from 1-bit sorting to multiple bits, we can rely
on a recursion with two parameters, $n$ and $K = 2^k$.
This two-parameter recursion trick is inspired by Lin et al.~\cite{osortsmallkey}:
however Lin et al.~\cite{osortsmallkey}
relies on randomness to perform divide-and-conquer and make measure  
concentration arguments in their proofs. 
To get rid of the randomness  
needed in their algorithm, we instead rely on the aforementioned linear-sized selection gadget
to perform the divide-and-conquer (this was not an available building 
block in Lin et al.'s work). 

Initially, we have $n$ elements each tagged with a $k$-bit key, i.e., in total
there can be $K = 2^k$ distinct keys.
We first rely on a linear-sized selection circuit to find the median.  
We then partition the input $n$ elements into two halves: those that are
smaller than or equal to the median, and all other elements. 
Note that to implement this partitioning in the circuit model, 
we will rely on our new linear-sized compaction circuit.
Now, a crucial observation is the following:
at least one of the halves, henceforth called the {\it good half},
will have at most $\ceil{K/2}$ distinct keys, while the other
{\it bad half} may still have up to $K$ distinct keys. 
We now recurse on each half to sort each half --- the good half
of the recursion adopts the parameters $n/2, \ceil{K/2}$, whereas
the bad half adopts the parameters $n/2, K$.
Thus we have the following recurrence where $T(n, K)$ denotes the 
circuit size for sorting $n$ elements each with $\log_2 K$-bit keys:
\begin{align*}
T(n,K) = T(\left\lceil \frac{n}{2}\right\rceil, K) + T(\left\lceil \frac{n}{2}\right\rceil, \ceil{K/2})
 + h(n,w+\log K) \cdot O(n (w + \log K) ),
\end{align*}
where $h(x,y):=\poly(\log^* x - \log^* y)$,
and the base cases are:
if $K = 1$ then $T(n, K) = O(n w) \cdot h(n,w+\log K)$;
moreover, if $n = O(1)$ then $T(n, K) = O(w+k)$.
Solving this recurrence, we will obtain that 
\[T(n, K) = O(n (w + \log K)) \cdot \log K \cdot h(n,w+\log K)\]

\subsection{Additional Related Work}
\label{sec:related}

\gnote{- In the related work section, there was an abrupt switch from discussing previous works on selection to discussing works on compaction. Could you explain briefly what is the connection between the compaction, selection problems and the problem of constructing ORAM? Was there any work that studied the circuit complexity of compaction and selection before this work, except in the context of building ORAMs?
}

\paragraph{Sorting.}
Lin, Shi, and Xie~\cite{osortsmallkey} 
considered the complexity of (randomized) oblivious sorting 
for small keys --- their work is in a randomized Oblivious RAM model of 
computation\footnote{A probabilistic Oblivious RAM is one 
that requires the distribution of the algorithm's memory access
patterns to be computationally or statistically close regardless of the input.}
and they show how to accomplish sorting in $O(n k \log \log n /\log k)$ time 
assuming that the keys to be sorted are only $k$ bits long  and the payload can fit
in a single memory word.
Their result cannot be directly converted the circuit model due to the use of randomness;
instead, they show how to interpret their result in a ``{\it probabilistic} 
circuit family'' model, i.e.,
for every input array, an overwhelming fraction of the circuits in 
the family can correctly sort it. 
In such a probabilistic circuit family model, 
they accomplish $O(n w \cdot k \log \log n/\log k)$ circuit size for $k$-bit keys, which
is asymptotically worse than our deterministic result 
when $k < 2^{\log\log n/(\log^*n)^{2 + \epsilon}}$.

\ignore{
who constructed {\it randomized} algorithms
that sort elements with small keys (and thus their result is not in a circuit model).
More specifically, their construction
leverages a carefully crafted two-parameter
recursion. 
}

It is also long known that using non-comparison-based
techniques, sorting can be accomplished on a RAM 
in $o(n \log n)$ time~\cite{sortlinearandersson,Kirkpatricksort,hansort00,hansort01,thorupsort}.
Unfortunately, all these algorithms inherently  
rely on input-dependent memory accesses and 
we know of no approach to convert these algorithms to the circuit model 
in an efficiency-preserving manner.
By contrast, in the circuit model, the non-comparison-based approach
has not been explored before, even though it was a natural open
problem since Knuth's textbook~\cite{knuthbook}, where it was shown that 
comparison-based techniques have an $\Omega(n \log n)$
lower bound even for sorting 1-bit keys.

Goodrich and Kosaraju~\cite{ppmsort}
show that on a parallel pointer machine, sorting 
$k$-bit integers in time $O(nk)$ using binary search trees. However, in their algorithm, 
the access patterns of the 
program are input data dependent, and thus their techniques do not 
directly lend to our result in the circuit model.


\paragraph{Compaction and selection.}
As mentioned, compaction is trivial to accomplish on a RAM; and linear-time 
selection on a RAM is also a classical result~\cite{blumselect} now widely taught
in undergraduate algorithms lectures. 
A line of work was concerned about the circuit complexity
of selection and compaction~\cite{yaoselectnet,sortminmem,jimboselect,pippengerselect};
but all earlier works focused on the comparator-based model.
Due to the famous 0-1 principle described as early as in Knuth's textbook~\cite{knuthbook},
there is an $\Omega(n \log n)$ lower bound for compaction  
with comparator-based circuits.
Furthermore, selection is also known to have an $\Omega(n \log n)$ lower bound 
in a comparator-based circuit model~\cite{yaoselectnet,sortminmem,jimboselect,pippengerselect} 
--- and this explains
why all earlier works~\cite{yaoselectnet,sortminmem,jimboselect,pippengerselect} 
focus on tightening the constant in front of the $n \log n$.
While it is extremely natural to ask whether the $n\log n$ barrier
can be circumvented without the comparator-based restriction, 
it is surprising that so far, no progress has been made along this front!

Several works have considered compaction and selection 
in other incomparable models of computation as explained below --- none of these results
easily give rise to a circuit result.
Leighton et al.~\cite{leightonselection}
show how to construct comparison-based, {\it probabilistic} circuit families for selection, 
with $O(n \log \log n)$ comparators; again, here we require that for every input, an overwhelming
fraction of the circuits in the family  
can give a correct result on the input. 
One interesting observation 
is that although deterministic, comparator-based {\it selection} circuits
have a $\Omega(n \log n)$
lower bound~\cite{leightonselection,sortminmem}, probabilistic circuit families 
do not inherit the same lower bound.
\elaine{why?}%
Subsequent works~\cite{osortsmallkey,odsmitchell} have improved Leighton's result by removing
the restriction that the circuit family must be parametrized
with $m$, i.e., the rank of the element being selected, without increasing the asymptotical
overhead.
These works also imply that compaction 
can be accomplished with in $O(n \log \log n)$
time on a randomized Oblivious RAM~\cite{osortsmallkey,odsmitchell}.

Very recently, Asharov et al. considered how to accomplish  
compaction on {\it deterministic} Oblivious RAMs in linear work~\cite{optorama}.
Their work was subsequently extended~\cite{paracompact}
to a PRAM setting.  Moreover, Dittmer and Ostrovsky improve
its concrete constants by introducing randomness back into the model~\cite{DittmerOstrovsky20}.
It turned out that linear-time oblivious compaction 
played a 
pivotal role in the construction of an optimal Oblivious RAM (ORAM) compiler~\cite{optorama}, 
a machine that translates a RAM program to a functionally-equivalent 
one with oblivious access patterns. 
Specifically, earlier ORAM compilers relied on oblivious sorting which requires 
$\Omega(n\log n)$ time 
either assuming the indivisibility model~\cite{osortsmallkey} or 
the Li-Li network coding conjecture~\cite{sortinglbstoc19}; whereas more recent
works~\cite{optorama,panorama} observed that with 
with a lot more additional work, we could replace oblivious sorting
with the weaker compaction primitive.



We stress that {\it known linear-time compaction algorithms~\cite{optorama,paracompact}
on Oblivious RAMs/PRAMs and our compaction circuit are 
of incomparable nature}: 
\begin{MyItemize}
\item 
First, 
an $O(n)$-time construction on a deterministic Oblivious RAM~\cite{optorama,paracompact} 
does not directly give rise
to an $O(n \cdot w)$-sized compaction circuit in general (assuming that 
$w = O(\log n)$ can be stored in a single memory word). This is because on a RAM,
operations on $\log n$-bit words can be accomplished with unit cost, whereas
in a boolean circuit they typically cost at least $\log n$ gates.
This is the case with the algorithms of Asharov et al.~\cite{optorama,paracompact} too. In fact, these algorithms achieve linear work by heavily and explicitly relying on the fact that operations on $\log n$-bit words can be accomplished with unit cost 
through techniques called {\it packing}.
It might be possible to use the same techniques as in Section~\ref{sec:firstattempt} to convert the 
linear-time compaction algorithms on Oblivious RAMs~\cite{optorama,paracompact} to the circuit model
--- indeed one could try this direction, but even if it works, it would  
result in a circuit of size  $O(n \log n+nw)$, i.e., no better than converting Pippenger's 
self-routing superconcentrator~\cite{selfroutingsuperconcentrator} to the circuit model.

\item 
Second, our $O(n w) \cdot \poly(\log^*n - \log^* w)$ compaction
circuit result also does not give rise to 
an Oblivious RAM algorithm running in time $O(n) \cdot \poly(\log^*n - \log^* w)$,
even if assuming that the word size is $w = O(\log n)$.
In order to see that, consider a circuit of size $O(w)$ boolean gates. One might hope to get an oblivious RAM counterpart that consumes $O(1)$ work---by packing all the $w$-wires of the circuit into a single word in the RAM model and performing operations on the packed word. However, in order to achieve $O(1)$ work, the operations on these $w$-wires in the circuit must fit one of the operations allowed in the RAM model, such as additions, multiplications, etc. Since boolean gates are usually more general and provide a different type of operations on each one of the wires individually, in the most general case we will still have to perform $w$ different operations in the RAM model even when packing the $w$-wires into one memory word in the RAM model. That is, the circuit may not have the SIMD (single instruction, multiple data) structure such that operations on $w$ bits can be packed into operations on the same work on a RAM. We also refer the readers to the work of Boyle and Naor~\cite{isthereoramlb} who discussed the same technical issue. 

\end{MyItemize}

Finally, sorting and selection have also been studied in the 
Parallel RAM (PRAM) model~\cite{parallelhash,parallelintegersorting,hanselection}
as well as in Valiant's comparison-based parallel 
computation model which charged only the comparisons but 
not the boolean computations that decided which elements to compare~\cite{comparetree,beateigen}.

\paragraph{Other related work.}
Besides Pippenger's self-routing 
super-concentrator\cite{selfroutingsuperconcentrator}, 
the work by Arora, Leighton, and Maggs~\cite{alm90}
considered a self-routing {\it permutation} network. Their construction  
is not a sorting network. 
Further, converting their non-blocking network to a permutation 
circuit would incur at least $\Omega(n \log^2 n)$ gates~\cite{bmm}.
Pippenger's work~\cite{selfroutingsuperconcentrator}
adopted some techniques 
from the Arora et al. work~\cite{alm90}.

%% file: figure.tex

\tikzstyle{swap}=[draw, fill=grey!10, text width=10em, 
    text centered, minimum height=2.5em, rounded corners, drop shadow]

\tikzstyle{looseCompaction} = [draw,     text centered, text width=10em, fill=green!60, 
    minimum height=2.5em, rounded corners, drop shadow]
\tikzstyle{tc} = [draw,     text centered, text width=10em, fill=green!20, 
    minimum height=2.5em, rounded corners, drop shadow]
\tikzstyle{tex} = [draw, fill=blue!20, text width=10em, 
    text centered, minimum height=2.5em, rounded corners, drop shadow]
\tikzstyle{mainTight} = [draw, fill=green!20, text width=10em, 
    text centered, minimum height=12em, minimum width=12em, rounded corners, drop shadow]
\tikzstyle{mainLoose} = [draw, fill=green!60, text width=10em, 
    text centered, minimum height=12em, minimum width=12em, rounded corners, drop shadow]
\def\blockdist{1}

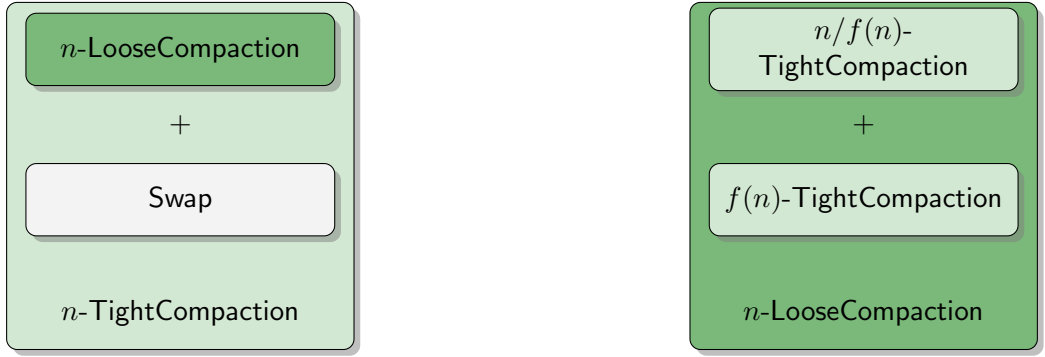
\begin{figure}
\begin{subfigure}[b]{0.45\textwidth}    
\begin{center}
\begin{tikzpicture}[
			auto, 	
			 .style={rectangle, draw, dotted, inner sep= 2pt, text width=2.2cm, text height=5.2cm},
			]
   
    \node (s) [mainTight] {};
    \path (s.south) + (0, 4*\blockdist)  node (looseCompaction) [looseCompaction] {$n$-{\sf LooseCompaction}};
    \path (s.south) + (0,3*\blockdist) node (tex){+}; 
    \path (s.south) + (0,2*\blockdist) node (swap) [swap]{${\sf Swap}$};
    \path (s.south) +(0,0.5*\blockdist) node (asrs) {$n$-{\sf TightCompaction}};
  
\end{tikzpicture}
%
\caption{\small A circuit for {\sf TightCompaction} (Section~\ref{sec:tightfromloose}) consists of {\sf Swap} (Section~\ref{sec:swapper}) and ${\sf LooseCompaction}$ (Section~\ref{sec:looseFromTight})}
\end{center}
\end{subfigure}%
\hfill%
\begin{subfigure}[b]{0.45\textwidth}    
\begin{center}
\begin{tikzpicture}[
			auto, 	
			 .style={rectangle, draw, dotted, inner sep= 2pt, text width=2.2cm, text height=5.2cm},
			]
   
    \node (s) [mainLoose] {};
    \path (s.south) + (0, 4*\blockdist)  node (tc) [tc] {$n/f(n)$-{\sf TightCompaction}};
    \path (s.south) + (0,3*\blockdist) node (tex) {+}; 
    \path (s.south) + (0,2*\blockdist) node (swap) [tc] {$f(n)$-{\sf TightCompaction}};
    \path (s.south) +(0,0.5*\blockdist) node (asrs) {$n$-{\sf LooseCompaction}};
  
\end{tikzpicture}
%
\caption{\small A circuit for {\sf LooseCompaction} (Section~\ref{sec:looseFromTight}) consists of a $f(n)$-${\sf TightCompaction}$ and ${n/f(n)}$-{\sf TightCompaction} }
\end{center}
\end{subfigure}%
\caption{\small Demonstrating the recursion and the dependency between {\sf LooseCompaction} and {\sf TightCompaction},
where $n$-{\sf X} denotes the {\sf X} circuit taking input size $n$.}
\label{fig:demonstration}
\end{figure}

%% file: prelim.tex
\section{Preliminaries}
\label{sec:preliminaries}

\paragraph{Notations.}
Unless otherwise noted, $\log$ always means $\log_2$.
Throughout the paper, the notation $\log^{(i)} n$ 
means $\log \log \log \ldots \log n$ where $\log$ is taken $i$ times;
the notation $\log^* n$ means the smallest $i \in \N$ such that  
$\log^{(i)} n \leq 1$.

\ignore{
\paragraph{Tight compaction.}
An $(n, w)$-tight compactor solves 
the following problem: given an input array 
consisting of $n$ elements each of bit-width $w$, where each element
is marked with an additional bit 
indicating whether the element is {\it distinguished}, 
output a permutation of the input array where 
all the distinguished elements are moved to the front 
and the non-distinguished elements are moved to the end. 
}


\elaine{should we define a version that truncates the output? otherwise
there are loose wires in the loose compact alg}

\subsection{Circuit Model of Computation}

We adopt the standard boolean circuit model of computation~\cite{savagebook}.
A boolean circuit consists of AND, OR, and NOT gates. Gates are connected through wires. 
Each gate's output can be fed into $O(1)$ number of other gates as input. 
All the gates and their wires form a directed acyclic graph.
By definition, such a boolean circuit has constant fan-in and constant fan-out.
A circuit's size is the number of gates 
the circuit contains.

\subsection{Operational Model Used in Our Paper}

For convenience, we use the following operational model for circuits when describing
our algorithms.
\ignore{
If there is a circuit of size $G(n)$ 
that computes a function  
in our operational model, there is also a standard boolean circuit 
consisting of $O(G(n))$ 
AND/OR/NOT gates that compute the same function.
}
In our operational model, we allow the following
gadgets in our circuitry: 1) constant-sized 
gadgets that implement an arbitrary truth table between the $O(1)$ input 
bits and the $O(1)$ output bits; and 2) selector gates.
When defining our circuits we will count explicitly how many generalized boolean gates the circuit has, and how many selector and reverse selectors are there. 
We give more explanations below and describe how to convert a circuit in our operational model 
to a standard boolean circuit with constant fan-in and constant fan-out.

\begin{MyItemize}
\item 
{\it Selector and reverse selector gates.}
A $w$-selector gate  (``MUX'') takes in one bit $b$ and two input elements $m_0,m_1$
each of bit-width $w$, and outputs $m_b$.
Each $w$-selector gate can be realized with 
$C \cdot w$ number of AND, OR, and NOT gates where $C > 1$ is a universal constant.

Reverse selector gates are the opposite. A $w$-reverse selector 
gate takes one 
element $m$ of bit-width $w$ and a signal $b \in \{0, 1\}$ as input and 
outputs $(m,0^w)$ if $b=0$ and $(0^w,m)$ if $b=1$. 

Henceforth in the paper whenever we count 
selector and reverse selector gates, we do not distinguish between them and count
both towards selector gates. 

\item 
{\it Generalized boolean gates.}
Generalized boolean gates are allowed to have constant fan-in and constant fan-out.
Without loss of generality, we allow 
each generalized boolean gate to implement {\it any truth table}
between the $O(1)$ input bits to the $O(1)$ output bits. 
Realizing such a generalized boolean gate 
with AND, OR, and NOT gates 
will only incur an $O(1)$ factor blowup in the circuit's size.
Later in our paper, we will use 
generalized boolean gates of fan-in at most 3 and fan-out at most
3. \elaine{adjust these constants in the end. they need to be universal constants} 

\ignore{
\item 
{\it Copy gates.}
Copy gates take a single input bit as input and can replicate the  
bit arbitrarily many times.
Copy gates come for free due to the following lemma. \elaine{need more explanation} 
}


\end{MyItemize}

\paragraph{Unbounded fan-out is free.}
From standard complexity theory textbooks~\cite{savagebook}, we
know that constant fan-in, unbounded fan-out circuits of size $G(n)$  
can be converted to constant fan-in, constant fan-out 
circuits of size $O(G(n))$. \elaine{double check the precise statement}
Therefore, in our operational model, we 
allow the output of a gate to be fed into unbounded number of gates.


\ignore{
\begin{lemma}
Given a circuit consisting of $L$ output wires, 
$S$ $w$-selector gates, 
$B$ boolean gates, and arbitrarily many input wires and copy gates, 
we can construct a functionally equivalent boolean circuit  
of constant fan-in and constant fan-out, whose
total size is upper bounded by $C \cdot (L + w S + B)$
for some universal constant $C > 1$.
\label{lem:circuitmodel}
\end{lemma}
\begin{proof}
\elaine{FILL}
\end{proof}
}


\begin{remark}[Selector gates and indivisibility model]
Although it may seem like selector gates are just another 
type of circuit gadgets
like those in Section~\ref{sec:gadgets}, we incorporate selector gates into our operational model
for convenience.

In fact, our construction later will use selector gates to move payloads 
around, and our construction will
not perform 
any boolean computation or encoding 
of the payloads.
This also means that our construction works in the 
{\it indivisibility} model.
In some works~\cite{isthereoramlb,osortsmallkey}, the indivisibility model is also referred
to as the {\it balls-and-bins} model, i.e., each element's payload
string behaves like an opaque, indivisible ball. 

Given a circuit in the indivisibility model, 
one construct a mirroring circuit in the \emph{reverse order}
where selector gates are converted to reverse-selector gates, and this reverse
circuit can route the payloads back into their initial positions.
Our construction later will make use of such ``reverse routing''.
\end{remark}


%% file: bldgblock.tex
\section{Useful Circuit Gadgets}
\label{sec:gadgets}

Besides the operational model described in Sections~\ref{sec:preliminaries}, we describe here some more complicated components for our circuits. Those includes some basic gadgets in Section~\ref{sec:basic-gadgets} (comparators, adders, counting, prefix sum), a circuit for converting a binary number to unary (Section~\ref{sec:binary-to-unary}), and much more complicated components: the basic (slightly inefficient) Loose Compaction circuit (Section~\ref{sec:lcprelim}) and Loose Swap (Section~\ref{sec:lswprelim}).

\subsection{Basic Gadgets}
\label{sec:basic-gadgets}

\paragraph{Comparator.}
A $k$-bit comparator takes two values each of $k$ bits, 
and outputs an answer from a constant-sized result set such as $\{>, <, =\}$, or
$\{\geq, <\}$, or $\{\leq, >\}$. 
Note that the outcome can be expressed as 1 to 2 bits.

\begin{fact}
A $k$-bit comparator can be implemented with $k$ generalized boolean gates.
\label{fct:comparator}
\end{fact}

\paragraph{Adder.}
A $k$-bit adder takes two values each of at most $k$ bits,
and outputs the sum of the two encoded as $k+1$ bits.
The following fact is obvious by emulating the hand method  
of addition.  

\begin{fact}
A $k$-bit adder can be implemented with $k$ generalized boolean gates
each with fan-in $3$ and fan-out $2$. 
\label{fct:adder}
\end{fact}

\paragraph{Counting.}
We consider a simple circuit gadget that counts the number
of $1$s in an input array containing $n$ bits.

\begin{fact}
Given an input array containing $n$ bits, counting the number of 1s in the input
array can be realized with 
a circuit containing $6n$ generalized boolean gates.
\label{fct:countprelim}
\end{fact}
\begin{proof}
To see this, consider a tree of adders. 
The input array forms the leaves of the tree where each leaf represents one bit.
At the leaf level there are $n/2$ adders adding $1$-bit numbers
and the outcome is at most $2$ bits;
at the next level, there are $n/4$ adders adding $2$-bit
numbers and the outcome is at most $3$ bits;
at the next level, there are $n/8$ adders adding $3$-bit numbers and
the outcome is at most $4$ bits, and so on.
By Fact~\ref{fct:adder}, 
the total number of generalized boolean gates needed is at most
\[
\frac{n}{2} \cdot 1 +  \frac{n}{4} \cdot 2 
+  \frac{n}{8} \cdot 3 + \ldots + 
1 \cdot \log_2 n \leq 6n
\]

\end{proof}

\paragraph{Prefix sum.}
We consider a prefix sum circuit gadget, which upon receiving
an input ${\bf I}$ containing $k$ bits, outputs the number of $1$s
encountered in all $k$ prefixes, that is,  ${\bf I}[:1], {\bf I}[:2], {\bf I}[:3], \ldots, {\bf I}[:k]$. This is implemented using a counter consisting of $\log k$ bits, and scanning the array and outputting the value of the counter each time. 
\begin{fact}
The aforementioned prefix sum can be computed with a circuit 
with at most $k \log k$ generalized boolean gates where $k$ is the input length. 
\label{fct:prefixsum}
\end{fact}

\elaine{TODO define a more general version}

\paragraph{Find in array.}
Given an input array containing $n$ elements 
each with a $k$-bit label and a $w$-bit payload,
given also a desired $k$-bit label $L^*$. 
Find the first occurrence in the input array an element whose label is $L^*$ 
and output any fixed canonical value if not found. 
This is implemented by $n$-iteration, compare the $k$-bit label of the current element in the array with $L^*$. The result of each comparison is used for a $w$-selector gate that selects whether to use the current $w$-bit payload or the one from previous iteration. 
\begin{fact}
The above task to find an element with desired $k$-bit label
in an array of length $n$ can be solved
with $nk$ 
generalized boolean gates and 
$n$ number of $w$-selector gates.
\label{fct:findinarray}
\end{fact}

\subsection{Binary to Unary Conversion}
\label{sec:binary-to-unary}
Fix $n \in \N$.
For any integer $k \in \{0,1,\dots,n\}$,
the binary-to-unary conversion takes as input $k$ and then outputs an $n$-bit binary string $u$,
where $k$ is represented in binary (string of $\ceil{\log n}$ bits),
and $u$ is the unary representation of $k$, i.e., 
the head $k$ bits of $u$ are all 0s and the tail $n-k$ bits of $u$ are all 1s.

Suppose $n$ is a power of 2
and $k$ is represented in $\log n + 1$ bits.
The following procedure implements binary to unary conversion.
Imagine that the $n$ output bits are at the leaf level of a complete binary tree.

\begin{mdframed}

\begin{myalgo}[Binary to Unary]
\label{alg:b_to_unary}
\end{myalgo}

\noindent
Initially, the root is labelled ``0'' if the most significant bit of $k$ is 1 (i.e., $k=n$);
otherwise, the root is labelled ``${\tt M}$'' denoting ``mixed''.
Henceforth consider the root to be at level $1$ of the tree.

\vspace{1mm}
\noindent
For each level $\ell$ from the root to 
the leaf (not including the leaf level): 
\begin{MyEnumerate}
\item 
if the node's label is not ``${\tt M}$'', simply pass its label
to both children;
\item 
else let $k_\ell$ denote the $(\ell+1)$-th bit of $k$:
\begin{MyItemize}
\item 
if $k_\ell = 1$, then pass ``0'' to the left child
and pass ``${\tt M}$'' to the right child;
\item 
if $k_\ell = 0$, then pass ``1'' to the right child
and pass ``${\tt M}$'' to the left child;
\end{MyItemize}
\end{MyEnumerate}

\noindent
Finally, if a leaf node receives the label ``${\tt M}$''
it is treated as ``1''.
\end{mdframed}

When $n$ is a power of 2, 
the above can be implemented
as a circuit consisting of at most $2n$  
generalized boolean gates. 
When $n$ is not a power of 2, 
we can simply skip a part of the tree in the above procedure such
that we propagate messages only to the first $n$ leave nodes;
thus the total number of 
generalized boolean gates is also upper bounded by $2n$.
We have the following fact.

\begin{fact}
The binary-to-unary conversion task 
can be implemented in a circuit with $2n$ generalized boolean gates.
\label{fct:binary_to_unary}
\end{fact}

\subsection{Loose Compaction}
\label{sec:lcprelim}
An $(n, w)$-loose compactor solves the following problem: 
\begin{MyItemize}
\item 
{\bf Input}:
an array containing $n$ elements 
of the form $\{(b_i, v_i)\}_{i \in [n]}$, 
where each $b_i \in \{0, 1\}$ is 
a metadata bit indicating 
whether the element is {\it real} or {\it dummy},
and each $v_i \in \{0, 1\}^w$ is the {\it payload}.
The input array is promised to {\it have at most $n/128$ real elements}.

\item 
{\bf Output}:
an array containing $n/2$ elements, 
such that the multiset of real elements contained in the output 
equals that of the input, i.e., 
no real element is lost or created. 
\end{MyItemize}

In other words, loose compaction takes a relatively sparse 
input array containing only a small constant fraction of real elements;
it compresses the input 
to half its original length while preserving 
all real elements in the input.

\begin{theorem}
There is a circuit in the indivisible model with $O(n \log n)$ generalized boolean gates 
and $O(n)$ number of $w$-selector gates that realizes
an $(n, w)$-loose compactor.
\label{thm:initlc}
\end{theorem}

The remainder of this subsection 
will be dedicated to prove Theorem~\ref{thm:initlc}.
We describe how to implement loose compaction as a circuit  
based on Pippenger's elegant ideas~\cite{selfroutingsuperconcentrator}. 
Specifically, we describe a slight variant of Pippenger's idea that appeared in
Asharov et al.~\cite{optorama}.

\paragraph{Expander graphs.}
The construction will rely on a suitable family of bipartite expander graphs
denoted $\{G_{\epsilon, m}\}_{m \in \N}$.
Specifically, $\epsilon \in (0, 1)$ is a suitable constant referred to as 
the {\it spectral expansion}. \elaine{what is this called?} 
The graph $G_{\epsilon, m}$
has $m$ vertices on the left henceforth denoted $L$, and $m$
vertices on the right henceforth denoted $R$,  
and each vertex has $d := d(\epsilon)$ number of edges
where $d$ is a constant that depends on $\epsilon$.

\paragraph{Construction.}
We formally describe the construction in Algorithm~\ref{alg:loose-compaction}, which is based on the one presented in~\cite{optorama}. We then present a circuit that implements this algorithm.

The main idea of the algorithm is to first distribute the real elements, such that when considering small chunks, no chunk is ``dense''. Then we can easily compact each chunk separately. The main technical challenge is in distributing the real elements. 
In more detail, the input array is grouped into chunks of $d/2$ size. 
Chunks that have at most $d/8$ elements (i.e., a quarter loaded)
are said to be {\it sparse} and chunks that have more than $d/8$ elements are said
to be {\it dense}.
We can consider the chunks to be left-vertices in the bipartite expander graph $G_{\epsilon, m}$.
Now each dense chunk will 
distribute its real elements to its neighbors on the right, such that each
right vertex receives no more than $d/8$ elements, i.e., each vertex
on the right is now a sparse chunk too.
At this moment, we can replace dense chunks on the left 
with dummy elements, as all real elements had moved to the right. At this point all chunks are sparse, and we can compress each chunk on the left and the right 
to a quarter its original size without losing any real elements. 
All compressed chunks are concatenated and output, and the output
array is a half the length of the input.  

The distribution of the real elements to its neighbors on the right requires some additional care, as we have to guarantee that no node on the right will become dense. We will have to compute on which subset of edges we will route the real elements. This is done via the procedure {\sf ProposeAcceptFinalize} described below.

\begin{figure}[h!]
\noindent
\begin{boxedminipage}{\textwidth}
\begin{center}

\begin{myalgo}[Loose Compaction]
\label{alg:loose-compaction}
\end{myalgo}
\vspace{-3ex}
\begin{MyItemize}
\item {\bf Input:} An input array ${\bf I}$ of $n$ elements, in which at most $n/128$ are real and the rest are dummies. 
\item
{\bf Assumption}:
Without loss of generality, we may assume that $d$ is a power of $2$ and 
$n$ is a multiple of $d/2$. Let $m := n/(d/2) = 2n/d$.


\item 
{\bf The algorithm}:
\begin{MyEnumerate}
\item 
Divide ${\bf I}$ into $m$ chunks of size $d/2$.
If a chunk contains at most $d/8$ real elements (i.e., at most a quarter loaded), 
it is said to be {\it sparse};
otherwise it is said to be {\it dense}.
It is not hard to see that the number of dense chunks must be at most  $m/32$.
\label{step:sparsedense}

\item 
Now imagine that each chunk is a vertex in $L$ of $G_{\epsilon, m}$, 
and $D \subset L$ is a set of dense vertices (i.e., corresponding to 
the dense chunks). 
Let ${\sf edges}(D, R)$ denote all the edges in $G_{\epsilon, m}$
between $D \subset L$ and $R$. 
 We would like to find a
subset of edges $M \subseteq {\sf edges}(D, R)$ 
 such that each vertex $u \in D$ has ${\sf load}(u) \leq d/2$ edges in $M$ and every vertex $R$ has at most $d/8$ incoming edges in $M$, where ${\sf load}(u)$ denotes the number of real elements in the chunk.
%
%
Henceforth we call such an $M$ a 
{\it feasible route}.

To find a feasible route given the set $D$, we rely on a subroutine 
called ${\sf ProposeAcceptFinalize}$ (see Algorithm~\ref{alg:propose-accept-finalize}).

\label{step:matchfinding}
\item 
Now, every dense chunk $u$ in $D$ 
does the following:
for each out edge of $u$ in $M$, send one element over the edge 
to a corresponding neighbor in $R$; for all out edges of $u$
not in $M$, send a dummy element on the edge.

Clearly every vertex in $R$ receives
no more than $d/8$ real elements.  
Henceforth we may consider every vertex in $R$
as a sparse chunk, i.e., an array of capacity $d/2$ but containing
only $d/8$ real elements. 

\label{step:onlineroute}
\item 
At this moment, scan through the vertices in $L$ and for each dense 
chunk encountered, replace the entire chunk with $d/2$ dummy elements.

\label{step:nodense}
\item 
Now, all chunks in $L$ and in $R$ must be sparse. That is, each chunk contains at most $d/8$ real elements, while its size is $d/2$. We now compress each chunk in $L$ and $R$ to a quarter of its original size (i.e., to size $d/8$ in length), without losing any real elements in the process. Let ${\bf O}$ denote the compressed array, containing of $2m \cdot \frac{d}{8} = 2 \cdot \frac{2n}{d} \cdot \frac{d}{8} = n/2$ elements. 

\label{step:compress}
\end{MyEnumerate}
\item {\bf Output:} The output array ${\bf O}$ of size $n/2$
\end{MyItemize}
\end{center}
\end{boxedminipage}
\end{figure}

\paragraph{{\sf ProposeAcceptFinalize} subroutine.}
We now describe the ${\sf ProposeAcceptFinalize}$ subroutine in Algorithm~\ref{alg:propose-accept-finalize},
which is the key step to achieve the aforementioned distribution of dense chunks.
To make the description more intuitive, henceforth
we call each vertex in $L$ a {\it factory} and each vertex in $R$ a {\it facility}.
Initially, 
imagine that the dense vertices correspond to factories that manufacture at most $d/2$
products, and the sparse vertices are factories that are unproductive.
There are at most $m/32$ productive factories, and they 
want to route all their products to facilities on the right satisfying the following
constraints: 1) each edge can route only 1 product; and 2) 
each facility can receive at most $d/8$ products.
The ${\sf ProposeAcceptFinalize}$ algorithm described below finds 
a set of edges $M$ to enable such routing, 
also called a feasible route as explained earlier. 

\medskip
\noindent

\begin{figure}[h]
\begin{boxedminipage}{\textwidth}
\begin{center}

\begin{myalgo}[{\sf ProposeAcceptFinalize} subroutine]
\label{alg:propose-accept-finalize}
\end{myalgo}
\vspace*{-3ex}
\end{center}
%
\noindent
Initially, each productive factory is {\it unsatisfied} and 
each unproductive factory is {\it satisfied}.
For a productive factory $u \in L$, we use notation ${\sf load}(u)$ 
to denote the number of products it has (corresponding to the number
of real elements in the chunk).

\noindent Repeat the following for 
$\log m$ times:

\begin{enumerate}[label=(\alph*),itemsep=0.1cm]
\item {\it Propose:} 
Each unsatisfied factory 
sends a proposal (i.e., the bit 1) to each one of its neighbors. 
Each satisfied factory sends 0 to each one of its neighbors. 

\item {\it Accept:}
If a facility $v \in R$
received no more than $d/8$ proposals, 
it sends an acceptance message to each one of its $d$ neighbors;
otherwise, it sends a reject message along each of its $d$ neighbors. 

\item  {\it Finalize:}
Each currently unsatisfied factory $u \in L$
checks if it received at least $\frac{d}{2}$ acceptance messages.
If so, it picks an arbitrary subset of the edges
over which acceptance messages were received, such that the subset
is of size ${\sf load}(u)$. 
The factory records these edges which are also added to the 
feasible route $M$.
At this moment, this factory becomes satisfied.
\end{enumerate}
\end{boxedminipage}
\end{figure}


\vspace{5pt}
Pippenger~\cite{selfroutingsuperconcentrator} and Asharov et al.~\cite{optorama}
have shown that if we use an appropriate family of bipartite expander graphs, the above
algorithm indeed realizes loose compaction correctly. More concretely, using properties of the expander graph, it is shown that in each iteration, the number of unsatisfied factories is decreased by a factor of at least $2$, and after $\log m$ iterations all factories are satisfied. We get the following proposition.

\begin{proposition}[Pippenger~\cite{selfroutingsuperconcentrator} and Asharov et al.~\cite{optorama}]
There exists an appropriate constant $\epsilon \in (0, 1)$
and a bipartite expander graph family $\{G_{\epsilon, m}\}_{m \in \N}$
where each vertex has 
$d := d(\epsilon)$ edges for a constant $d$ that is dependent on $\epsilon$ 
and assumed to be a power of 2, such that for any $m$ and $n = md/2$, 
the above loose compaction algorithm, when instantiated with this family of bipartite expander
graph,  can correctly compress any input array of length $n$ 
to a half its original size without losing any real elements, 
as long as the input array has at most $n/128$ real elements.
\label{prop:lc}
\end{proposition}

It remains to 
show how to implement the above algorithm in the circuit model.

\paragraph{Implementing loose compaction (Algorithm~\ref{alg:loose-compaction}) in circuit.}
We now consider how to implement the loose compaction algorithm (Algorithm~\ref{alg:loose-compaction}) in circuit.
\begin{MyEnumerate}
\item 
Step~\ref{step:sparsedense}
can be accomplished with $O(n)$ generalized boolean 
gates due to Fact~\ref{fct:countprelim} 
and Fact~\ref{fct:comparator}.

\item 
Step~\ref{step:matchfinding}
is the invocation of ${\sf ProposeAcceptFinalize}$ algorithm (Algorithm~\ref{alg:propose-accept-finalize}) that finds a feasible route. We will show how to implement this algorithm in circuit below, using $O(d \log d \cdot m \log m)$ generalized boolean gates, which is upper bounded by $O(m \log m)$ assuming that $d=O(1)$.


\item 
Recall that at the end of the ${\sf ProposeAcceptFinalize}$ subroutine, 
each factory has ``written down'' (i.e., these values were output by some gates
in the circuit) the bit-vector $\beta \in \{0, 1\}^d$ 
indicating whether each of its incident edges should bear load (this is representation of $M$ in circuit per node). 
Now, each factory, holding a $d/2$-sized chunk, wants to ``send'' 
every real element in the chunk on one of its out edges marked as $1$ by the bit-vector
$\beta$. 

To achieve this, we do the following per factory: 
\begin{MyItemize}
\item We 
label each real element in the $d/2$-sized chunk 
with a number indicating how many
real elements have appeared to its left, including itself. Every dummy element is labelled $\bot_1$. Using Fact~\ref{fct:prefixsum} this can be achieved in $O(d \log d)$.
\item 
Similarly, we label each 1-position 
in $\beta$ with a number indicating how many 
1-bits have appeared to its left, including itself. Each 0-position in $\beta$ is labelled with $\bot_2 \neq \bot_1$. Using Fact~\ref{fct:prefixsum} this can be achieved in $O(d \log d)$.
\item 
Now, each position in the vector  
$\beta$ grabs an element from the chunk whose label 
is the same as itself, and if not found, the canonical outcome is dummy. 
This can be accomplished using Fact~\ref{fct:findinarray} for $d$ times, resulting in $O(d^2 \log d)$ generalized boolean gates and $O(d^2)$ number of $w$-selector gates. 
\end{MyItemize}
In total over all factories, we have $O(m \cdot d^2 \log d)$ generalized boolean gates and $O(m \cdot d^2)$ number of $w$-selector gates over all factories.

\item 
Step~\ref{step:nodense} can be implemented with 
$m$ number of $(wd/2)$-selector gates, or alternatively
$m d/2$ number of $w$-selector gates.

\item 
Step~\ref{step:compress}
can be implemented with $(d/8) \cdot (d/2)$
number of $w$-selector gates 
and $(d/8) \cdot (d/2) \cdot \log d$ number of  
generalized boolean gates per chunk: 
basically imagine each 
input element is labeled with the number of real elements
to its left including itself (within the chunk). 
Each of the $d/8$ output positions
wants to find an input element whose label matches its own index.
We can implement the above using the prefix sum 
circuit of Fact~\ref{fct:prefixsum}
and $d/8$ copies of the find-in-array circuit of Fact~\ref{fct:findinarray}.

\ignore{
basically for each of the $d/8$ output positions $i \in [d/8]$, 
do the following to determine its value.
Sequentially scan
through the input chunk: every step along the way,  
the following can be computed using only the dummy indicator
of each input element (without looking at the element's payload):
\begin{enumerate}
\item the total number of real elements encountered so far denoted $s$; and 
\item 
whether $s = i$. 
\end{enumerate}
The latter boolean output
will be fed into a selector gate  
that decides whether to replace the output position $i$'s value
with the current input element encountered.
\elaine{rewrite this more simply}
}

Summing over all $m$ chunks, Step~\ref{step:compress} requires a total of  
$m \cdot (d/8) (d/2) \cdot \log d$ generalized boolean gates and 
$m \cdot (d/8) (d/2)$ number of $w$-selector gates.
\end{MyEnumerate}

\paragraph{Implementing ${\sf ProposeAcceptFinalize}$ (Algorithm~\ref{alg:propose-accept-finalize})  as a circuit.}
Note that the circuit's wiring structure can encode
the expander graph $G_{\epsilon, m}$ and sending messages
over the edges of the graph is carried out by the circuit's wiring,
which feeds the outputs of circuit gadgets corresponding to factories (or facilities resp.)
to the inputs of circuit gadgets  
corresponding to facilities (or factories resp.).
Since there are $\log m$ iterations, the expander graph is in fact
encoded $\log m$ times in the circuit's wiring.
More specifically, for each of the $\log m$ iterations:
\begin{MyItemize}
\item  
Imagine that each facility has a circuit gadget that counts how many   
of its incoming wires have the bit $1$, compares the outcome with $d/8$, and 
outputs an accept/reject decision.
By Fact~\ref{fct:comparator} and Fact~\ref{fct:countprelim}, 
this step can be implemented with 
$O(d)$ 
generalized boolean gates per facility and thus
in total  $O(d m)$
generalized boolean gates.

\item 
Now each factory uses a circuit gadget 
that reads the accept/reject decisions on each of its  
$d$ incoming wires, tallies the total acceptance messages,  compares
the outcome to $d/2$, and then updates its satisfied/unsatisfied indicator.\footnote{Note
that when implemented in circuit, every update to a variable creates
a new copy of the variable that is output by some gate(s). 
In other words a circuit can be thought of as a straightline 
program of a ``single-assignment form''~\cite{ssa} where the names
of wires are variable names in the program. \label{ftnt:ssa}}
This step can be implemented with $O(d m)$
generalized boolean gates counting all factories 
by Fact~\ref{fct:comparator} and Fact~\ref{fct:countprelim}.

\item 
We want that at the end of the 
${\sf ProposeAcceptFinalize}$ algorithm, each factory writes
down a bit-vector $\beta = \{0, 1\}^d$ of length $d$
indicating whether each of its outgoing edge is chosen.
To make this possible, each factory uses a circuit gadget to update
its bit-vector $\beta \in \{0, 1\}^d$ at the end of each iteration as follows:
if the factory was unsatisfied before but became satisfied in this iteration (this
can be determined by looking at the old value of the satisfied label
and the new value, see Footnote~\ref{ftnt:ssa}), 
then update the bit-vector 
$\beta = \{0, 1\}^d$, setting $\beta[i] = 1$ iff the $i$-th incoming wire
has an acceptance decision and moreover 
the number of acceptance decisions until the $i$-th incoming wire
is not more than the number of products it wants to route. 
This step can be accomplished with $O(d \log d)$
generalized boolean gates per factory due to Facts~\ref{fct:comparator} and \ref{fct:prefixsum}.
Thus over all factories we have a total of $O(m d \log d)$ 
generalized boolean gates.
 \end{MyItemize}

Over all $\log m$ iterations, the total number of generalized
boolean gates needed is $O(d \log d \cdot m \log m)$, which
is upper bounded by $O(m \log m)$ assuming that $d = O(1)$.

Summarizing the above, we  have the following fact: 
\begin{corollary}
The above loose compaction algorithm can be implemented 
as a circuit with $O(n \log n)$
generalized boolean gates and $O(n)$ number of $w$-selector gates.
\label{cor:lccost}
\end{corollary}

Theorem~\ref{thm:initlc}
now follows directly from Proposition~\ref{prop:lc}
and Corollary~\ref{cor:lccost}.

\subsection{Loose Swap}
\label{sec:lswprelim}
A loose swapper obtains an input array where
each element is marked with $\bot$, ${\tt blue}$, or ${\tt red}$,
and moreover the number of ${\tt red}$ elements is the 
same as the number of ${\tt blue}$ elements.
Now, a loose swapper circuit swaps a subset of the ${\tt blue}$
elements with ${\tt red}$ ones and the swapped elements receive $\bot$.
Henceforth we call elements marked ${\tt red}$ or ${\tt blue}$
{\it colored} and those marked $\bot$ {\it uncolored}.

Formally, an $(n, w)$-loose swapper solves the following problem:
\begin{MyItemize}
\item {\bf Input}:
an input array containing $n$ elements 
where each element contains a $w$-bit payload string
and a two-bit metadata label  
whose value is chosen  
from the set $\{{\tt blue}, {\tt red}, \bot\}$.
Henceforth we assume the first bit of the label
encodes whether the element is colored or not, and the second bit of the label
picks a color between ${\tt blue}$ and ${\tt red}$ if the element is indeed colored. 

\item {\bf Output}:
a {\it legal swap} of the input array such that 
at most $n/128$ elements remain colored, where the notion of
a legal swap is defined below.

We say that an output array ${\bf O}$ is a {\it legal swap}
of the input array ${\bf I}$ iff  
there exist pairs $(i_1, j_1), (i_2, j_2), \ldots, (i_\ell, j_\ell)$
of indices that are all distinct, 
such that for all $k \in [\ell]$, ${\bf I}[i_k]$
and ${\bf I}[j_k]$ are colored and have opposite colors, 
and moreover ${\bf O}$ is obtained by swapping 
${\bf I}[i_1]$ with ${\bf I}[j_1]$, 
swapping ${\bf I}[i_2]$ with ${\bf I}[j_2]$,  $\ldots$, and 
swapping ${\bf I}[i_k]$ with ${\bf I}[j_k]$; further, all swapped
elements become uncolored. 

\elaine{this constant matches that of loose compact}
\end{MyItemize}

\begin{theorem}
There exists a circuit in the indivisible model with $O(n)$
generalized boolean gates and $O(n)$ number of $w$-selector gates
that realizes an $(n, w)$-loose swapper.
\label{thm:lswap}
\end{theorem}
\begin{proof}
We will describe an algorithm first described by Pippenger~\cite{selfroutingsuperconcentrator}
and then re-explained by Asharov et al.~\cite{optorama}.
The algorithm makes use of a suitable bipartite expander graph $G_{\epsilon, n}$
where $\epsilon \in (0, 1)$ is a suitable constant.
The degree of each vertex in $G_{\epsilon, n}$
is a constant $d := d(\epsilon)$ that depends on $\epsilon$.
Henceforth let $L$ denote the set of left-vertices in $G_{\epsilon, n}$
and $R$ denote the set of right vertices; and $|L| = |R| = n$.

We will think of every element in the input array as a vertex in $L$. 
During the algorithm, each vertex $u_1 \in L$ performs the following 
actions acting in a sequential manner:
\begin{itemize}
\item 
For each vertex in $u_2 \in L$ that is connected to $u_1$ by a length-2 path,  
if $u_1$ and $u_2$ currently have opposite colors, then swap them and uncolor the two vertices.
\end{itemize}

Pippenger~\cite{selfroutingsuperconcentrator}
and Asharov et al.~\cite{optorama}
show that there exists a suitable bipartite expander graph $G_{\epsilon, n}$
with a constant spectral expansion $\epsilon \in (0, 1)$, 
\elaine{is this the name}
such that if the above algorithm is performed using the graph $G_{\epsilon, n}$, 
it will correctly realize a loose swapper on $n$ elements.

Note that every vertex in $L$ has at most $d^2$ two-hop neighbors in $L$.
Since the circuit's wiring can encode the expander graph's structure,
it is not hard to see that we can implement the above algorithm in a circuit, incurring 
$n \cdot d^2$ comparisons on 2-bit color labels, 
and $2 n \cdot d^2$
number of $w$-selector gates.
\end{proof}

%% file: tightfromloose.tex

\section{Tight Compaction from Loose Compaction}
\label{sec:tightfromloose}
In this section, we show how to construct a circuit for tight compaction from Loose Compaction and from a Swapper circuit. This corresponds to case (a) in Figure~\ref{fig:demonstration}. We remark that we already saw a basic Loose Compaction circuit in Section~\ref{sec:lcprelim}, but in our final construction we will iteratively improve that circuit via bootstrapping and boosting. 

\begin{theorem}
Suppose that there is a circuit with $B_{\rm lc}(n)$ 
generalized boolean gates and $S_{\rm lc}(n)$ $w$-selector gates
that loosely compacts an input array containing $n$ elements each of bit-width $w$.
Suppose also that there is a loose swapper circuit with 
$B_{\rm lsw}(n)$ 
generalized boolean gates and $S_{\rm lsw}(n)$ $w$-selector gates
for an input array containing $n$ element each of bit-width $w$. 
Moreover, suppose that $B_{\rm lc}(n) \geq n$, 
$S_{\rm lc}(n) \geq n$, 
$B_{\rm lsw}(n) \geq n$, and
$S_{\rm lsw}(n) \geq n$. 

Then, tight compaction can be accomplished with a circuit 
with at most
$
2S_{\rm lsw}(n) + 
2B_{\rm lsw}(n) + 2B_{\rm lc}(n) + 8S_{\rm lc}(n) + 11n$
generalized boolean gates, and
at most
$2S_{\rm lsw}(n) + 4S_{\rm lc}(n) + 2n$
number of $w$-selector gates.  
\label{thm:tightcompact}
\end{theorem}


The remainder of this section will be dedicated to proving this theorem.

\subsection{From Loose to Tight Compaction Algorithm}
\label{sec:tightcompact}
Consider the following algorithm --- below we first describe the high-level algorithm
and then we will describe how to implement each step in circuit. Assume that in the input array ${\bf I}$ there are $c$ distinguished elements. Thus, the distinguished elements should be placed in the first $c$ elements in the output array. Moreover, all distinguished elements in the first $c$ positions in the input array ${\bf I}$ are already in the ``right place'' and all non-distinguished elements in the last $n-c$ positions of ${\bf I}$ are also in the right place and should not be moved. Moreover, there are exactly the same number of non-distinguished elements in the first $c$ positions in ${\bf I}$ as the number of distinguished elements in the last $n-c$ positions in ${\bf I}$. The tight compaction algorithm will simply swap them. In more detail, the algorithm works as follows:

\begin{MyEnumerate} 
\item 
{\it Count.}
Compute the total number (denoted $c$) of distinguished elements (i.e., real elements) 
in the input array ${\bf I}$. 
\label{step:count}
\item 
{\it Color.}
For any $i \leq c$, if ${\bf I}[i]$ is not distinguished, mark the element ${\tt red}$; 
for any $i > c$, if ${\bf I}[i]$ is 
distinguished, mark the element ${\tt blue}$; every other element is marked $\bot$.
Let the outcome be ${\bf X}$. 

Note that at this moment, each element 
is labeled with 3 bits of metadata, one bit of distinguished indicator
and two bits of color-indicator (indicating whether the element is colored,
and if so, which color).
\label{step:mark}

\item 
{\it Swap.}
Call an $(n, w + 1)$-swapper (see Section~\ref{sec:swapper}) to swap 
each ${\tt blue}$ element with a ${\tt red}$ ones (we use here payload of size $w+1$ and not $w$ as we also include the color-indicator as part of the payload). 
Specifically, an $(n, w + 1)$-swapper 
is defined exactly like a loose swapper but  
with the requirement that the outcome array must have no colored elements remaining.
\end{MyEnumerate}

In the next couple of subsections we will 
focus on explaining how to realize each of the above steps in circuit.
The most non-trivial step 
is the swapper which leverages
a loose swapper and a loose compactor as a building block.
Thus we will first describe the swapper circuit
(Section~\ref{sec:swapper})
and then explain how to implement 
the remaining steps in circuit (Section~\ref{sec:remainingsteps}).

\subsection{Swapper Circuit}
\label{sec:swapper}

We now focus on how to realize an $(n, w)$-swapper in circuit. We first describe the algorithm in  Algorithm~\ref{alg:swap} and then explain how to realize it as a circuit. 

\noindent 
\begin{figure}[h]
\begin{boxedminipage}{\textwidth}
\begin{myalgo}
[{${\sf Swap}({\bf X})$}]
\label{alg:swap}
\end{myalgo}
\vspace{-3ex}
\begin{MyItemize}
\item {\bf Input:} An array {\bf X} of $n$ elements, each has a $w$-bit payload\footnotemark and a 2-bit label indicating whether the element is colored,
and if so, whether the element is ${\tt blue}$ or ${\tt red}$.
\item {\bf The algorithm:}
\begin{MyEnumerate}
\item 
Call an $(n, w)$-loose swapper (see Section~\ref{sec:lswprelim}) on ${\bf X}$ to swap elements of opposite colors 
and uncolor them in the process,
such that at most $1/128$ fraction of resulting array remain colored. 

If $n \leq 128$, output the resulting array; else continue with the following steps. 

\label{step:swap}
\item 
Call an $(n, w+1)$-loose compactor (see Section~\ref{sec:looseFromTight})
to compact the outcome of the previous step by a half, 
where the loose compactor treats the colored elements 
as real and the uncolored elements as dummy.
In other words, the loose compactor treats the 1st bit of the color label
as a dummy indicator, and treats the 2nd bit of the color label and an 
element's  payload string as the payload.
Let the outcome be ${\bf Y}$ whose length is half of ${\bf X}$.
\label{step:compact}
\item 
Recursively call ${\sf Swap}({\bf Y})$, and let the outcome be ${\bf Y}'$.
\item 
Reverse the routing decisions made by all selector gates 
during Step~\ref{step:compact} as below.
For every selector gate $g$ in Step~\ref{step:compact}, 
its reverse selector gate denoted $g'$ 
is one that receives a single element as input and outputs two elements; the same control
bit $b$ input to the original gate $g$ 
is used by $g'$ to select which of the output 
receives the input element, and the other one will simply receive the string $0^{w+1}$.
If $g$ selected the first input element to route to the output, then 
in $g'$, the input element is routed to the first output.

In this way, we can reverse-route elements in ${\bf Y}'$ 
to an array (denoted $\widetilde{\bf X}$) of length $|{\bf X}|$, i.e., twice the length
of ${\bf Y}'$.

\label{step:reverse}
\item 
The output ${\bf Z}$ is formed by a performing coordinate-wise 
select operation between ${\bf X}$ and $\widetilde{\bf X}$:
\[
{\bf Z}[i] = 
\begin{cases}
{\bf X}[i] & \text{if ${\bf X}[i]$ is uncolored}\\
\widetilde{\bf X}[i] & \text{o.w.}\\
\end{cases}
\]
\ignore{back to the input ${\bf X}$, such that each position 
$i$ receives either a real or a dummy element: if a real element
is received the element ${\bf X}[i]$ is overwritten with the received element;
otherwise the element ${\bf X}[i]$ is unchanged.
Finally, output the result.
}
\label{step:output}
\end{MyEnumerate} 
\item {\bf Output:} The  array {\bf Z}. 
\end{MyItemize}
\end{boxedminipage}
\footnotetext{Our tight compaction algorithm in Section~\ref{sec:tightcompact} actually
requires a swapper where elements are of bit-length $w+1$, but for convenience
we rename the variable to $w$ in the description of the swapper.}
\end{figure}

\paragraph{Implementing Algorithm~\ref{alg:swap} in circuit.} This swapper is a recursive construction 
that is executed on arrays of length $n, n/2, n/4, \ldots$.
For each length $n'$, we consume a loose compactor, a loose swapper, 
and a reverse-router (accompanying
the loose compactor) for the size $n'$. 
Thus for each problem size $n' = n, n/2, n/4, \ldots$, we need 
\begin{MyItemize}
\item 
$S_{\rm lsw}(n')$ number of $w$-selector gates 
and 
$B_{\rm lsw}(n')$ number of generalized boolean gates
corresponding to Step~\ref{step:swap}; 
\item 
$2 S_{\rm lc}(n')$
number of $(w+1)$-selector gates 
(one for the forward direction
and one for the reverse direction)
and $B_{\rm lc}(n')$
generalized boolean gates 
 corresponding to Step~\ref{step:compact}; 
and  
\item 
$n'$ number of $w$-selector gates corresponding to Step~\ref{step:output}. 
\end{MyItemize}

Note that each $(w+1)$-selector gate can be realized 
with one $w$-selector gate that operates on the $w$-bit payload 
and one generalized boolean gate that computes on the extra metadata bit. 
Thus each problem size $n'$ can be implemented with 
$S_{\rm lsw}(n') + 2 S_{\rm lc}(n') + n'$
number of $w$-selector gates and 
$B_{\rm lsw}(n') +  B_{\rm lc}(n') + 2 S_{\rm lc}(n')$
generalized boolean gates.

Summing over all $n' = n, n/2, n/4, \ldots$, 
and recalling that $B_{\rm lc}(n) \geq n$, 
$S_{\rm lc}(n) \geq n$, $B_{\rm lsw}(n) \geq n$, and
$S_{\rm lsw}(n) \geq n$, we have the follow fact:

\begin{fact}
In the swapper circuit shown above which operates on elements of bit-width $w$, 
the total number of $w$-selector gates 
needed is upper bounded by 
$2S_{\rm lsw}(n) + 4 S_{\rm lc}(n) + 2n$
and the total number of generalized boolean gates
is upper bounded by 
$2B_{\rm lsw}(n) +  2B_{\rm lc}(n) + 4S_{\rm lc}(n)$.
\label{fct:swap}
\end{fact}

\subsection{Implementing the Remaining Steps in Circuit}
\label{sec:remainingsteps}
We now describe how to implement the above algorithm with a circuit ---
without loss of generality, we may assume that $n$ is a power of $2$:

\paragraph{Step~\ref{step:count}: counting.}
Due to Fact~\ref{fct:countprelim}, 
the following is immediate:

\begin{fact}
Step~\ref{step:count}
can be accomplished with $6n$ generalized boolean gates
and no selector gate.
\label{fct:count}
\end{fact}

\ignore{
Step~\ref{step:count} can be achieved using a tree of adders like below.
Imagine that every real element is marked with $1$ and every dummy
marked with $0$. These form the leaves of the tree.
At the leaf level there are $n/2$ adders adding $1$-bit numbers 
and the outcome is at most $2$ bits;
at the next level, there are $n/4$ adders adding $2$-bit 
numbers and the outcome is at most $3$ bits; 
at the next level, there are $n/8$ adders adding $3$-bit numbers and
the outcome is at most $4$ bits, and so on.
It is not hard to see that adding two $\ell$-bit numbers can be implemented
$\ell$ generalized boolean gates each with fan-in $3$ and fan-out $2$
(basically emulating the hand method of addition).
Thus, the total number of generalized boolean gates needed is at most
\[
\frac{n}{2} \cdot 1 +  \frac{n}{4} \cdot 2 
+  \frac{n}{8} \cdot 3 + \ldots + 
1 \cdot \log_2 n \leq 6n
\]
\elaine{double check the above expression}
}

\paragraph{Step~\ref{step:mark}: coloring.}
When the outcome $c$ is computed from Step~\ref{step:count}, we can implement
Step~\ref{step:mark} as follows. Recall that $c \in \{0,1, \dots,n\}$ 
is a $(\log_2 n)$-bit number.
Imagine that there are $n$ receivers numbered 
$1, 2, \ldots, n$. Each receiver is waiting to receive  
either ``$\leq$'' or ``$>$''.
Those with indices $1, \ldots, c$ should receive ``$\leq$'' and 
those with indices $c+1, \ldots, n$ should receive ``$>$''.
Using Fact~\ref{fct:binary_to_unary},
we convert $c$ into an $n$-bit string so that the head $c$ bits are 0
and the tail $n-c$ bits are 1. 
Such $n$ bits are passed to the $n$ receivers where 0 is interpreted as ``$\leq$''
and 1 is interpreted as ``$>$'',
and the above can be implemented
as a circuit consisting of at most $2n$  
generalized boolean gates.
Once each of the $n$ receivers receive either ``$\leq$'' or ``$>$'',  
it takes a single generalized boolean gate 
per receiver (with fan-in 2 and fan-out 2) 
to write down either ${\tt blue}$, ${\tt red}$, or $\bot$. 

Therefore, the total number of generalized boolean 
gates needed for this step is upper bounded by $3n$; and no selector gates are needed here.

\begin{fact}
Step~\ref{step:mark}
can be accomplished with $3n$ generalized boolean gates
and no selector gate.
\label{fct:color}
\end{fact}

\subsection{Putting it Together}
Summarizing Facts~\ref{fct:swap}, \ref{fct:count}
and \ref{fct:color}, 
for the entire tight compaction algorithm of Section~\ref{sec:tightcompact}, 
we need at most $2B_{\rm lsw}(n) + 2B_{\rm lc}(n) + 4S_{\rm lc}(n) + 9n$
generalized boolean gates, and 
at most 
$2S_{\rm lsw}(n) + 4S_{\rm lc}(n) + 2n$
number of $(w+1)$-selector 
gates. Since each $(w+1)$-selector gate can be replaced
with one $w$-selector gate and 
one generalized boolean gate, 
alternatively we can realize tight compaction in circuit
with at most
$ 
2S_{\rm lsw}(n) + 
2B_{\rm lsw}(n) + 2B_{\rm lc}(n) + 8S_{\rm lc}(n) + 11n$
generalized boolean gates, and
at most
$2S_{\rm lsw}(n) + 4S_{\rm lc}(n) + 2n$
number of $w$-selector gates which gives rise to the statement
in Theorem~\ref{thm:tightcompact}.

%% file: loosefromtight.tex

\section{Loose Compaction from Tight Compaction}
\label{sec:looseFromTight}
In this section, we show how to construct a circuit for loose compaction from tight compaction. This corresponds to case (b) in Figure~\ref{fig:demonstration}. 


\begin{theorem}
Let $f(n)$ be some function 
in $n$ such that $1 < f(n) \leq \log_2 n$ holds for every $n \geq 32$; let $C_{\rm tc} > 1$
be a constant. Suppose that 
$(n,w)$-tight compaction can be solved
by a circuit 
with $C_{\rm tc} \cdot n \cdot f(n)$ 
generalized boolean gates and $C_{\rm tc} \cdot n$ selector gates for any $n \geq 256$, 
then loose compaction can be solved by a circuit 
with $2.07 C_{\rm tc} \cdot n \cdot f(f(n)) + 7n$ boolean
gates and $2.07 C_{\rm tc} \cdot n$ selector gates for any $n \geq 256$.
\elaine{double check statement, i changed the constants to 256}
\label{thm:loosecompact}
\end{theorem}

The remainder of this section will be dedicated to proving the above theorem.


\subsection{Loose Compaction Algorithm}
\label{sec:loosecompact}
For simplicity, we first consider the case when $n$ is divisible
by $f(n)$. Looking ahead, we will use $f(n)$ to be $\log^{(x)}n$ for some $x$ that is power of $2$. 
We will later extend our theorem statement 
to the case when $n$ is not divisible by $f(n)$.  
Consider the following algorithm:

\medskip

\noindent 
\begin{boxedminipage}{\columnwidth}
\begin{myalgo}
[{${\sf Loose Compaction from Tight Compaction}$}]
\label{alg:looseCompactionFromTightCompaction}
\end{myalgo}
\vspace{-3ex}
\begin{MyEnumerate}
\item 
Divide the input array into $f(n)$-sized chunks. 
We say that a chunk is {\it sparse} if there are at most $f(n)/32$ real elements in it; otherwise
it is called {\it dense}.
Now, count the number of elements in every chunk, 
and mark each chunk as either ${\tt sparse}$ or ${\tt dense}$.
\label{step:countperchunk}
\item 
Call an $(n/f(n), w\cdot f(n))$-tight compactor \elaine{todo: use this notation elsewhere}
 to move  
the dense chunks to the front and the sparse chunks  
to the end.
\label{step:movechunks}

\item 
We will show later in Fact~\ref{fct:sparsechunk} 
that at least $3/4$ fraction of the chunks are sparse.
Now, apply a $(f(n), w)$-tight compactor 
to the trailing $\ceil{\frac{3}{4} \cdot \frac{n}{f(n)}}$ chunks to compress
each of these chunks  
to a length of $\floor{\frac{f(n)}{32}}$
without losing any elements in the process.
The first $\frac{1}{4}\cdot \frac{n}{f(n)}$
chunks are unchanged. Output the resulting array.
\elaine{note: now there are loose-end output wires... use the other notion of tight compaction}

\label{step:compactperchunk}
\end{MyEnumerate}
\end{boxedminipage}

\medskip
At the end of the algorithm, the output array has length at most
\begin{equation}
\frac{3}{4} \cdot \frac{n}{f(n)}
\cdot \frac{f(n)}{32} 
+ \frac{1}{4} \cdot \frac{n}{f(n)} \cdot f(n) \leq 0.28 n < 0.5 n 
\label{eqn:compressionrate}
\end{equation}
\elaine{note: the output is not 1/2 of original, it's shorter} 

\begin{fact}
At least ${\frac{3}{4} \cdot \frac{n}{f(n)}}$ chunks are sparse.
\label{fct:sparsechunk}
\end{fact}
\begin{proof}
Suppose not, this means that more than  
${\frac{1}{4} \cdot \frac{n}{f(n)}}$
have more than $f(n)/32$ real elements.
Thus the total number of elements is more than 
$n/128$ \elaine{note more than} 
which contradicts the input sparsity assumption of loose compaction.
\end{proof}

\subsection{Implementing Algorithm~\ref{alg:looseCompactionFromTightCompaction} in Circuit}
\label{sec:loosecompactcircuitsize}
We now analyze the circuit size of the algorithm in Section~\ref{sec:loosecompact}.
For simplicity, we first assume that $n$ is divisible by $f(n)$
and we will later modify our analysis to the more general case when $n$ is not
divisible by $f(n)$.

\begin{MyEnumerate}
\item 
Step~\ref{step:countperchunk}
of the algorithm 
requires 
at most $6n$ generalized boolean gates as we have $n/f(n)$ counters of chunks of size $f(n)$ each. Due to Fact~\ref{fct:count}, each counter requires at most $6f(n)$ generalize boolean gates. 


\item 
Step~\ref{step:movechunks} is a single invocation of a $(n/f(n),w\cdot f(n))$-tight compactor. Assuming that $(n,w)$-tight compactor can be realized with $C_{\rm tc}\cdot n \cdot f(n)$ generalized boolean gates and $C_{\rm tc} \cdot n$ selector gates, this step requires at most 
$C_{\rm tc} \cdot (n/f(n)) \cdot f(n/f(n)) \leq C_{\rm tc} \cdot (n/f(n)) \cdot f(n)$
generalized boolean gates and $C_{\rm tc} \cdot n/f(n)$
number of $w\cdot f(n)$-selector gates.
Each such selector gate can in turn be realized with $f(n)$ number of $w$-selector gates.
Thus, in total, 
Step~\ref{step:movechunks}
requires $C_{\rm tc} \cdot n$ generalized boolean gates
and $C_{\rm tc} \cdot n$ number of $w$-selector gates. 

\item 
Step~\ref{step:compactperchunk}
of the algorithm requires applying 
$\ceil{\frac{3}{4} \cdot \frac{n}{f(n)}}$
number of 
$(f(n), w)$-tight compactors, where, according to our assumption in Theorem~\ref{thm:loosecompact}, each such tight compactor 
consumes $C_{\rm tc} \cdot f(n) \cdot f(f(n))$ generalized boolean
gates and $C_{\rm tc} \cdot f(n)$ number of $w$-selector gates. 
For $n \geq 32$ and $f(n) \leq \log_2 n \leq n/4$, we have that 
\[
\ceil{\frac{3}{4} \cdot \frac{n}{f(n)}} \cdot f(n) \leq 
\left(\frac{3}{4} \cdot \frac{n}{f(n)} + 1\right) \cdot f(n) 
\leq 3n/4 + n/4 \leq n 
\]
Therefore, 
in total there are at most 
$C_{\rm tc} \cdot n \cdot f(f(n))$ generalized boolean
gates and $C_{\rm tc} \cdot n$
number of $w$-selector gates.
\elaine{the constant here can be made 0.76}
\end{MyEnumerate}

\begin{fact}
Assume the same assumptions as in Theorem~\ref{thm:loosecompact}, and moreover
$n$ is divisible by $f(n)$.
The loose compaction algorithm in Section~\ref{sec:loosecompact}
can be realized with a 
circuit consisting of $C_{\rm tc} \cdot n \cdot (f(f(n)) + 1) + 6n$
generalized boolean gates 
and $2 C_{\rm tc} \cdot n$ number of $w$-selector gates.
\end{fact}

\paragraph{When $n$ is not divisible by $f(n)$.}
When $n$ is not divisible by $f(n)$, we can pad the last chunk with dummy elements to a 
length of $f(n)$.  After the padding the total number of elements is upper bounded
by $n + f(n)$. This gives rise to the following fact.

\begin{fact}
Assume the same assumptions as in Theorem~\ref{thm:loosecompact}.
Then, the loose compaction algorithm in Section~\ref{sec:loosecompact}
can be realized with a 
circuit consisting of $2.07 C_{\rm tc} \cdot n \cdot f(f(n)) + 7n$
generalized boolean gates 
and $2.07 C_{\rm tc} \cdot n$ number of $w$-selector gates.
\end{fact}
\begin{proof}
Recall that we padded the input array with dummy elements
 to a length that is a multiple of $f(n)$. The number of padded elements 
is at most $f(n)$.
We first check that Equation~\ref{eqn:compressionrate}
still holds, i.e., the algorithm 
compresses the input array by at least a half.
Note that the input padded array still satisfies the $1/128$ input sparsity assumption.
The output array now has length 
upper bounded by 
\[
\frac{3}{4} \cdot \frac{n + f(n)}{f(n)}
\cdot \frac{f(n)}{32} 
+ \frac{1}{4} \cdot \frac{n + f(n)}{f(n)} \cdot f(n) \leq 0.28 (n + f(n)) < 0.5 n 
\]
Note that the last inequality above holds because $n$ is sufficiently large.

Now, repeating the analysis above in Section~\ref{sec:loosecompactcircuitsize}.
we can show that Algorithm~\ref{alg:looseCompactionFromTightCompaction}
can be realized with a 
circuit consisting of $C_{\rm tc} \cdot (n + f(n)) \cdot (f(f(n)) + 1) + 6(n + f(n))$
generalized boolean gates 
and $2 C_{\rm tc} \cdot (n + f(n))$ number of $w$-selector gates.
To obtain the expression in the above fact, it suffices to observe
that for $n \geq 256$ and $f(n) \leq \log_2 n$, it holds
that 
\[
\frac{f(f(n))+1}{f(f(n))} \leq 4/3, \quad \frac{n+f(n)}{n} \leq 33/32 \ .
\]
\end{proof}

\ignore{note: after the padding, we need to make sure 
that 0.28 * (n + f(n) ) < 0.5 n.
the expression becomes  
$C_{\rm tc} \cdot (n + f(n)) \cdot (f(f(n) + 1) + 6 (n + f(n))$
generalized boolean gates 
and $2 C_{\rm tc} (n + f(n)) $ number of $w$-selector gates.
}

%% file: mainresult.tex

\section{Linear-Sized Tight Compaction Circuit}
\label{sec:lineartc}
In this section, we shall prove the following theorem. 
We will tighten the constant-degree of the $\poly$ to $2+\epsilon$ in Appendix~\ref{sec:tighten}.

\begin{theorem}[Linear-sized tight compaction]
There exists a constant fan-in, constant fan-out boolean circuit 
that solves $(n, w)$-tight compaction and the total
number of boolean gates is upper bounded by 
$$
 O(n \cdot w) \cdot \min\left(\max\left(1,\poly(\log^* n - \log^* w)\right),2^w/w\right).
$$

As a direct corollary, for any arbitrarily large constant $c \geq 1$, 
if $w \geq \log^{(c)} n$, it holds that the circuit's size
is upper bounded by $O(n w)$.
\label{thm:main}
\end{theorem}

The case when $w > \log n$ is easy (see Footnote~\ref{fnt:bigw}), 
so in the remainder
of this section, unless otherwise noted, 
we shall assume that $w \leq \log n$.
In this section, we prove only that the number of total boolean gates is upper bounded by $\poly(\log^* n - \log^* w) \cdot O(n \cdot w)$. We refer the reader to Appendix~\ref{sec:tinyw} for the second part, which dominants when $w$ is tiny. In Appendix~\ref{sec:tighten} we tighten the constant-degree of the $\poly$. 
%


\subsection{Notations and Parameter Choices}
In Theorem~\ref{thm:lswap}, we showed the existence of a $(n,w)$-loose swapper that works in $O(n)$ generalized boolean gates and $O(n)$ $w$-selector gates. Henceforth we write $O_1(n) = c n$ and $O_2(n) = c' n$ where
$c$ and $c'$ are universal constants. Specifically, 
let 
\[4 B_{\rm lsw}(n) + 11n \leq  O_1(n) = cn , \quad 
2 S_{\rm lsw}(n) + 2n \leq O_2(n) = c'n 
\]
where $B_{\rm lsw}(n)$ and $S_{\rm lsw}(n)$
are linear functions in $n$ by Theorem~\ref{thm:lswap}. Thus, we have the following circuit:
\begin{itemize}
\item[$\LC_0$:]
By Theorem~\ref{thm:initlc}, there exists
a constant $C > 1$ 
such that we can solve $(n, w)$-loose compactor with
\[
\begin{array}{rl}
\text{\rm generalized boolean gates}: &  C n \log n
\\
\text{\rm selector gates}: &  C n \\
\end{array}
\] 
\end{itemize}

Without loss of generality, we may assume
that the constant $C$ is sufficiently large
such that the following expressions hold: 
\begin{equation}
7n \leq 0.03 \cdot 4.1Cn,  
\quad 
O_2(n) \leq 0.1 Cn
\label{eqn:largeC}
\end{equation}

\paragraph{Additional notations.} Recall that $\log$ means $\log_2$.
We will choose the depth of the recursion $d$ to be the smallest positive 
integer 
such that 
$\log^{(2^{d-1})} n \leq w$.
Without loss of generality, we may redefine
$\log n := \max(w, \log n)$ --- 
due to the choice of $d$, 
essentially for the last recursion level, 
if $\log^{(2^{d-1})} n < w$, we will round it up to $w$; moreover, this rounding
is only performed for the last level of recursion and no other level. 
Therefore, we may assume that 
\begin{equation}
\log^{(2^{d-1})} n = w
\label{eqn:depth}
\end{equation}

Without loss of generality, we may 
assume that the bit-length of an element $w$ is lower bounded by a sufficiently large
constant, such that the following expressions are satisfied:
\begin{equation}
w \geq \frac{8 C n + O_1(n)}{2.1 C n}, \quad 
w \geq 256
\label{eqn:largew}
\end{equation}



\subsection{Construction through Repeated Bootstrapping and Boosting}
\label{sec:repeated_bootstrap}
We will construct tight compaction through 
repeated bootstrapping and boosting.
Without loss of generality, we may assume that $n \geq 256$.
\elaine{check the condition on n}
\noindent
We have two steps:
\begin{MyItemize}
\item {\bf \boldmath $\LC_i \Longrightarrow \TC_{i+1}$ (Theorem~\ref{thm:tightcompact}):}
from loose compactor to tight compactor. Simplifying Theorem~\ref{thm:tightcompact} and using the above notations, we get the following:
\begin{quote}
Assuming $(n,w)$-loose compactor with:
\[
\begin{array}{rl}
\text{\rm \# generalized boolean gates}: &  B_{\rm lc}(n) 
\\
\text{\rm \# selector gates}: &  S_{\rm lc}(n) \\
\end{array}
\] 
Then, there exists $(n,w)$-tight compactor with:
\[
\begin{array}{rl}
\text{\rm \# generalized boolean gates}: &  2B_{\rm lc}(n) + 8S_{\rm lc}(n) + O_1(n)
\\
\text{\rm \# selector gates}: &  4S_{\rm lc}(n)+O_2(n)
\end{array}
\] 
\end{quote}
\item {\bf \boldmath $\TC_{i+1} \Longrightarrow \LC_{i+1}$ (Theorem~\ref{thm:loosecompact})}: from tight compaction to loose compactor. Simplifying  Theorem~\ref{thm:loosecompact} we have:
\begin{quote}
Assuming $(n,w)$-tight compactor with some constant $C_{\rm tc}$ and function $f(n)$ such that:
\[
\begin{array}{rl}
\text{\rm \# generalized boolean gates}: &  C_{\rm tc}\cdot n\cdot f(n)
\\
\text{\rm \# selector gates}: &  C_{\rm tc}\cdot  n 
\end{array}
\] 
Then there exists a $(n,w)$-loose compactor such that:
\[
\begin{array}{rl}
\text{\rm \# generalized boolean gates}: &  2.07\cdot C_{\rm tc}\cdot n\cdot f(f(n))
\\
\text{\rm \# selector gates}: &  2.07\cdot C_{\rm tc} \cdot n 
\end{array}
\] 

\end{quote}
\end{MyItemize}

\noindent

Our starting point is Theorem~\ref{thm:initlc}, which as we have already seen, it gives as the circuit $\LC_0$. Using the above two steps, we bootstrap and boost the circuit:
\begin{itemize}
\item[$\LC_0$:]
By Theorem~\ref{thm:initlc}, there exists
a constant $C > 1$ 
such that we can solve $(n, w)$-loose compactor with
\[
\begin{array}{rl}
\text{\rm generalized boolean gates}: &  C n \log n
\\
\text{\rm selector gates}: &  C n \\
\end{array}
\]

\item[$\TC_1$:] By Theorem~\ref{thm:tightcompact}, we can 
construct a tight compaction circuit $\TC_1$ from $\LC_0$.
$\TC_1$'s size is upper bounded by the 
expressions\footnote{\label{fnt:bigw} When 
$w > \log n$, $\TC_1$ gives Theorem~\ref{thm:main}. 
Therefore, the rest of this section assumes $w \leq \log n$.
}:
\[
\begin{array}{rl}
\text{\rm generalized boolean gates}: &  2 C n \log n + 8 C n + O_1(n) \leq 4.1 C n \log n 
\\
\text{\rm selector gates}: &  4 C n + O_2(n) \leq 4.1 C n \\
\end{array}
\] 
In the above, the first the inequality holds 
due to Equations~(\ref{eqn:depth})
and (\ref{eqn:largew}) as $8 C n + O_1(n) \leq 2.1Cnw$ and $w \leq \log n$.
The second inequality holds
due to Equation~(\ref{eqn:largeC}).

\item[$\LC_1$:] 
Using the algorithm in Section~\ref{sec:loosecompact}, we build
a loose compaction circuit $\LC_1$ from $\TC_1$.
 $\LC_1$'s size is upper bounded by the expressions:
\[
\begin{array}{rl}
\text{\rm generalized boolean gates}: &  
2.07 \cdot 4.1 C n \log \log n + 7 n \leq 2.1 \cdot 4.1 Cn \log \log n \\
\text{\rm selector gates}: &   2.07 \cdot 4.1 Cn  \leq 2.1 \cdot 4.1 Cn\\
\end{array}
\] 
In the above, the first inequality holds
due to Equation~(\ref{eqn:largeC}).

\item[$\TC_2$:] Using the algorithm in Section~\ref{sec:tightcompact}, we can 
construct a tight compaction circuit $\TC_2$ from $\LC_1$.
$\TC_2$'s size is upper bounded by the expressions:
\[
\begin{array}{rl}
\text{\rm generalized boolean gates}: &  
2 \cdot 2.1 \cdot 4.1 Cn \log \log n  + 8 \cdot 2.1 \cdot 4.1 Cn + O_1(n) \leq 
2.1 \cdot (4.1)^2 C n \log \log n
\\
\text{\rm selector gates}: &  
4 \cdot 2.1 \cdot 4.1 Cn  + O_2(n) \leq 2.1 \cdot (4.1)^2 Cn 
\\
\end{array}
\] 

In the above, the first the inequality holds 
due to Equations~(\ref{eqn:depth})
and (\ref{eqn:largew}) as $8 C n + O_1(n) \leq 2.1Cn \log\log n$.
The second inequality holds
due to Equation~(\ref{eqn:largeC}).

\item[$\LC_2$:] 
Using the algorithm in Section~\ref{sec:loosecompact}, we build
a loose compaction circuit $\LC_2$ from $\TC_2$.
 $\LC_2$'s size is upper bounded by the expressions:
\[
\begin{array}{rl}
\text{\rm generalized boolean gates}: &  
2.07 \cdot 2.1 \cdot (4.1)^2 C n \log^{(4)}n + 7 n \leq (2.1 \cdot 4.1)^2 Cn \log^{(4)}n 
\\
\text{\rm selector gates}: &   
2.07 \cdot 2.1 \cdot (4.1)^2 Cn  \leq (2.1 \cdot 4.1)^2 Cn
\\
\end{array}
\] 
\end{itemize}

\noindent
Continuing with the iterations for $d$ iterations, we get:
\begin{itemize}
\item[$\LC_{d-1}$:] 
Using Theorem~\ref{thm:loosecompact}, we build
a loose compaction circuit $\LC_{d-1}$ from $\TC_{d-1}$.
 $\LC_{d-1}$'s size is upper bounded by the expressions:
\[
\begin{array}{rl}
\text{\rm generalized boolean gates}: &  
(2.1 \cdot 4.1)^{d-1} Cn \log^{(2^{d-1})}n 
\\
\text{\rm selector gates}: &   
(2.1 \cdot 4.1)^{d-1} Cn
\\
\end{array}
\] 

\item[$\TC_d$:]
Using Theorem~\ref{thm:tightcompact} we can 
construct a tight compaction circuit $\TC_d$ from $\LC_{d-1}$.
$\TC_d$'s size is upper bounded by the expressions:
\[
\begin{array}{rl}
\text{\rm generalized boolean gates}: &  
(2.1)^{d-1} \cdot (4.1)^d C n \log^{(2^{d-1})}n
\\
\text{\rm selector gates}: &  
(2.1)^{d-1} \cdot (4.1)^d Cn 
\\
\end{array}
\] 
\end{itemize}

By definition, $d - 1 = \ceil{\log(\log^* n - \log^* w)} \leq \log(\log^* n - \log^* w) + 1$.
Therefore, the final tight compaction circuit $\TC_d$'s size is upper bounded by the following:
\begin{align*}
\text{\# generalized boolean gates:} & \leq 
O(1) \cdot (2.1 \cdot 4.1)^{\log(\log^* n - \log^* w) + 1} \cdot n \cdot w\\
& = 
O(1) \cdot {\poly(\log^* n - \log^* w)} \cdot n \cdot w \\
\text{\# selector gates:} & \leq 
O(1) \cdot (2.1 \cdot 4.1)^{\log(\log^* n - \log^* w) + 1} \cdot n \\
& = 
O(1) \cdot {\poly(\log^* n - \log^* w)} \cdot n\\
\end{align*}
This gives rise to Theorem~\ref{thm:main}.
For the case of a tiny $w$ we can upper bound the size of the circuit by $2^w/w$, as we show in Appendix~\ref{sec:tinyw}.

%% file: selection.tex

\section{Linear-Sized Selection Circuit}


\ignore{
Given our compaction circuit, we can construct a selection circuit 
whose size is asymptotically the same. 
We first consider the variant 
where we only want to select the $m$-th smallest element;
but we can easily extend the algorithm to select all $m$ smallest elements.
}

We care about selecting all $m$ smallest elements from an input array, using a linear-sized circuit.
Given our tight compaction circuit, it suffices to select only the $m$-th smallest
element $x$ since after that, we can use tight compaction to move all elements
smaller than or equal to $x$ to the left.
To select the $m$-th smallest element with a linear-sized circuit,
it suffices to combine our linear-sized tight compaction circuit
with the classical, textbook median-of-median algorithm~\cite{blummedian}.
Essentially, the step in which the median-of-median algorithm partitions
elements according to a pivot will be replaced with our tight compaction circuit. 
\ignore{
We first consider how to select the $m$-th smallest element from an input array.
Given our earlier linear-sized compaction circuit, 
if we can select the $m$-th smallest element with a linear-sized circuit,
then we can also easily select 
all $m$ smallest elements with a linear sized circuit.
To select the $m$-th smallest element from an input array,
we can run the classical median-of-median algorithm~\cite{blummedian} but 
now encoding the algorithm with 
a boolean circuit which uses compaction circuits as a building block. 
}
The algorithm is described in Algorithm~\ref{alg:select}.

\begin{figure}[h]
\noindent \begin{boxedminipage}{\textwidth}
\begin{myalgo}[${\bf Select}({\bf A}, m)$]
\label{alg:select}
\end{myalgo}
\vspace{-3ex}
\begin{MyEnumerate} 
\item 
Let $n := |{\bf A}|$. If $n \leq 100$, use the AKS
sorting network to sort the input array and 
output its $m$-th element; else continue with the following steps.

\item 
Divide the elements into $\ceil{n/5}$ groups each of size $5$. 
If $n$ is not divisible by $5$, the last group may have fewer than $5$ elements.
\item 
Compute the median of each group. Note that 
this can be accomplished with an $O(w)$-sized circuit where $w$ is the bit-width
of each element.
Let ${\bf M}$ be the array of $\ceil{n/5}$ medians.
\item 
Recursively call ${\bf Select}({\bf M}, 1/2 \cdot \abs{{\bf M}})$
to compute the median of the median, and let $m^*$ be the outcome.
\item  
Mark each element with $0$ if it is smaller than $m^*$ and $1$ otherwise.
Use a compaction circuit to move all the elements marked with 0 to the left, and those
marked with 1 to the right.
Let ${\bf X}$ be the resulting array.
\ignore{
\item  
Scan through ${\bf X}$ and identify $r := {\sf rank}(m^*)$, i.e.,
the index of $m^*$ within ${\bf X}$ --- if there are multiple
occurrences of $m^*$, pick the one whose rank is closest to $n/2$.  \elaine{is this ok}
Note that this can be done with an $O(n w)$-sized circuit. 
}
\item 
Let $n' = \floor{7n/10 + 3}$; 
let $c_{\rm small}$ be the number of elements strictly smaller than $m^*$ 
and let $c_{\rm big}$ be the number of elements strictly greater than $m^*$
--- note that $c_{\rm small}$ and $c_{\rm big}$ can be computed in an $O(n)$-sized circuit.
\item 
Depending on $m$, recursively call ${\bf Select}$ on an array  
of size $n'$ as follows:
\begin{MyItemize}
\item 
if $m \leq c_{\rm small}$, 
recursively call ${\bf Select}({\bf X}[1:n'], m)$ and output its outcome;
\item 
else if $m \geq n - c_{\rm big} + 1$, 
recursively call ${\bf Select}({\bf X}[n-n'+1:n], m - n + n')$ and output its outcome;
\item 
else 
recursively call ${\bf Select}({\bf X}[1:n'], m)$ and output $m^*$. 
\end{MyItemize}
(To implement the above as a circuit, first use a selector gate that  
selects between ${\bf X}[1:n']$ and ${\bf X}[n-n'+1:n]$, and then  
run ${\bf Select}$ recursively on the outcome.)
\end{MyEnumerate} 
\end{boxedminipage}
\end{figure}

\vspace{10pt}
For correctness, observe that during each iteration in the above recursive construction,   
there are at least 
$(\ceil{n/10} - 1)\cdot 3$ elements smaller than or equal to $m^*$;
similarly, there are at least 
$(\ceil{n/10} - 1)\cdot 3$ elements greater than or equal to $m^*$.

The circuit size for the above construction, denoted $C(n)$, satisfies the following recurrence:
\begin{align*}
\text{for $n > 100$:} & \ \  C(n) \leq C(\ceil{n/5}) + C( \floor{7n/10 + 3} ) + 
O(n w) \cdot \min(\poly(\log^*n - \log^* w), \ 2^w/w)\\
\text{for $n \leq 100$:}  & \ \ 
C(n) \leq C_0 \text{ for an appropriate constant $C_0$} 
\end{align*}

This recurrence solves to 
$C(n) = O(n w) \cdot \min(\max(1,\poly(\log^*n - \log^* w)), \ 2^w/w)$.

Finally, to select not just the $m$-th smallest element but all $m$ smallest elements, 
one can select the $m$-th smallest element denoted $m^*$ 
first and then rely on a compaction circuit
to move all elements smaller than or equal to $m^*$
to the left, and all other elements to the right.
Summarizing the above, we have the following corollary:

\begin{corollary}
There exists a circuit 
that can 
select the $m$-th smallest element or all $m$ smallest elements from an input array containing
$n$ elements each of bit-width $w$, 
and moreover its size is upper bounded by 
\[
O(n w) \cdot \min\left(\max(1, \poly (\log^*n - \log^* w)), \ \ 2^w/w \right)
\] 
As a direct corollary, if $w \geq \log^{(c)} n$
for any arbitrarily large constant $c \geq 1$ or if $w = O(1)$, 
then the circuit size is upper bounded by $O(n w)$.

\end{corollary}

As we show in Appendix~\ref{sec:tighten}, the constant-degree of the $\poly$ can
be as small as $2+\epsilon$ for an arbitrarily small constant $\epsilon > 0$.

\section{Sorting Elements with Small Keys}

Compaction is a sort with $1$-bit keys. In this section we show how to sort $w$-bits elements with keys of some range $[a+1, a+K]$ for some constants $a$ and $K$. We let $k = \log K$ and $W = w + k$ be the size of each element, where $w$ is the payload and $k$ is the size of the key. Our circuit has size of $(w+k)\cdot O(n k)\cdot h(n,w+k)$, 
where $h(x,y):=\max\left(1,\poly(\log^* x - \log^* y)\right)$ for short. 
\medskip

\begin{figure}[h]
\noindent \begin{boxedminipage}{\textwidth}
\begin{myalgo}[${\bf Sort^{[K]}}({\bf A})$]
\label{alg:sort}
\end{myalgo}
\vspace{-3ex}

\begin{MyItemize}
\item {\bf Input:} An array ${\bf A}$ containing elements with keys from domain $[a+1, a+K]$,
where each key is represented in $\ceil{\log K}$ bits as the shared offset $a$ is stored only once.
\item {\bf The algorithm:}
\begin{MyEnumerate} 
\item
(Base case.) 
If $K \leq 2$, run tight compaction on ${\bf A}$ and output the result directly.
Otherwise, proceed with the following step.
\item Find the median in the array by running ${\bf Select}({\bf A},n/2)$  (Algorithm~\ref{alg:select}). 
\item Count the number of elements that are
a) less than the median and b) exactly the median.
\item Mark exactly $n/2$ elements that are smaller than or exactly the median (using the counts of the previous step), and then run tight compaction (see Section~\ref{sec:lineartc}). At this point, all elements that are smaller than the median are in the first half of the array, and all elements that are greater or equal to the median are in the second half of the array. 
\item One of the two halves must have no more than $\lceil K/2 \rceil$ distinct keys (see argument regarding correctness). We find the minimum and the maximum key of each half, and thus can see the range of keys in each half. We call the half with less number of keys as the \emph{good half} and the half with more keys as the \emph{bad half}. If both contain the same number of keys, we break ties arbitrarily. 

\item Using a selector, we move the good half to the beginning of the array. Let $(G,B)$ be the two halves after the re-routing. 
\item We next recurse on the bad half using $\widetilde{B} \leftarrow {\bf Sort}^{[K]}(B)$, and on the good half by using $\widetilde{G} \leftarrow {\bf Sort}^{\left[\lceil K/2 \rceil\right]}(G)$. 
\item Reverse route $\widetilde{G}, \widetilde{B}$ (using a selector). 

\end{MyEnumerate}
\item {\bf Output:} The array ${\bf A}$. 
\end{MyItemize}
\end{boxedminipage}
\end{figure}

\medskip
\noindent
For correctness, we claim that after performing the partitioning, the following hold: 
Either the top half or the bottom half can still have $K$ distinct  keys remaining but not both; 
Moreover, one of the halves must have no more than $\lceil K/2 \rceil$ distinct keys remaining. 
This is true as one of the halves will have some $1 \leq x \leq K$ distinct keys, and so the other half would have $K+1-x$ distinct keys. 

As for the running time, in each iteration we run {\bf Select} and then tight compaction, and then selectors. Moreover, we have two recursive calls on the two halves of the array. We get: 
\begin{eqnarray*}
T(n,K) =  T(\left\lceil \frac{n}{2}\right\rceil, K) + T(\left\lceil \frac{n}{2}\right\rceil, K/2)
  + h(n,w+k) \cdot O(n  (w+k))
\end{eqnarray*}
Moreover, we have the following base cases:
\begin{eqnarray*}
{\rm For~} n \in O(1) &: & T(n,K) \in O(w+k) \\
{\rm For~} K \leq 2 & :& T(n,K) \in h(n,w+k) \cdot O(n \cdot (w+k))
\end{eqnarray*}
It is not hard to see that this recursion results in $T(n,K) \in O(n (w+k)\cdot k) \cdot h(n,w+k)$. Henceforth, we change the ${\sf poly}(x)$ in $h$ with $x^{2+\epsilon}$ as in Appendix~\ref{sec:tighten}. For summary, we have:
\begin{corollary}
There exists a circuit that can sort an array of $n$ elements of bit-width $w$, each element is marked with a key in $[a+1,a+K]$, and moreover its size is upper bounded by
$$
O((w+k) \cdot n k) \cdot \max\left(1,(\log^* n - \log^* (w+k))^{2+\epsilon}\right)
$$
where $k = \log K$ for an arbitrary small constant $\epsilon > 0$. 
As a direct corollary, if $w + k\geq \log^{(c)} n$
for any arbitrarily large constant $c \geq 1$, 
then the circuit size is upper bounded by $(w+k)\cdot O(n k)$.
\end{corollary}


\begin{proof}
Let $h(x,y) := \max(1,(\log^* x - \log^* y)^{2+\epsilon})$.
We show that 
$T(n,K) \leq C \cdot n(w+k)\cdot k \cdot h(n,w+k)$ for some constant $C$. 
We prove this claim by induction, while the constant $C$ is the max between the two base cases ($n \in O(1)$ and $K \leq 2$), and the inductive step. Assume that the induction holds to some constant $C$ for every $n' < n$, we want to show that $T(n',K) \leq C \cdot n'(w+k)\cdot k \cdot h(n',w+k)$. We want to prove that this holds also for $n'=n$. We have
\begin{eqnarray*}
T(n,K) & = &T(\frac{n}{2}, K) + T( \frac{n}{2}, K/2) + n\cdot (w+k)\cdot h(n,w+k) \\
 & \leq & C \cdot \frac{n}{2}\cdot (w+k)\cdot k \cdot h(\frac{n}{2},w+k)\\
& & ~~~+~ C \cdot \frac{n}{2}\cdot (w+k)\cdot (k-1) \cdot h(\frac{n}{2},w+k-1) \\
& & ~~~+~ n\cdot (w+k)\cdot h(n,w+k) 
\end{eqnarray*}
Thus,
\begin{eqnarray*}
T(n,K) & \leq & C \cdot n\cdot (w+k)\cdot k \cdot h(n,w+k) \\
& & -~ \frac{1}{2} C \cdot n \cdot (w+k)\cdot h(\frac{n}{2},w+k-1)\\ 
& & +~ n\cdot (w+k)\cdot h(n,w+k)
\end{eqnarray*}
If $w+k > \log n$, we have $h(n,w+k)=1$, the induction holds directly.
Otherwise, letting $X = \log^* n - \log^* (w+k)$, it holds that
\begin{eqnarray*}
\log^*\frac{n}{2} - \log^* (w+k-1) \leq (X-1).
\end{eqnarray*}
Then, the inductive step holds if 
\begin{eqnarray*}
\frac{1}{2} C \cdot n \cdot (w+k)\cdot (X-1)^{2+\epsilon}  \ge  n \cdot (w+k)\cdot X^{2+\epsilon}  \ , 
\end{eqnarray*}
which holds for every $C \geq 12$ and $n \geq 2^4$. 
\end{proof}

%% file: lb.tex

\section{Lower Bound}
\elaine{TODO: add some discussion about the lower bound stealing bits from payload}

Lin, Shi, and Xie~\cite{osortsmallkey} showed that any circuit in the 
indivisible model
that sorts $n$ elements 
each with a $k$-bit key
must have at least $\Omega(n k)$ selector gates.
Recall that a circuit in the indivisible model considers the elements'
payloads as opaque; such a circuit only moves the payloads around using selector gates
but does not perform any boolean computation on the payloads. 
Note that our upper bound is indeed in the indivisible model.
In this section, we show that 
a similar lower bound holds even without restricting
the circuit to satisfy the indivisibility assumption, assuming that  
a famous network coding conjecture to be true~\cite{lilinetcoding}.
Although our proof techniques are inspired the recent works 
of Farhadi et al.~\cite{sortinglbstoc19}
and Afshani et al.~\cite{lbmult}, we need additional non-trivial modifications
to make the techniques work in our context.

\elaine{maybe should discuss the point that the proof steals some bits from the payload}

\subsection{Preliminaries: The Li-Li Network Coding Conjecture}
Our lower bound is conditional and 
relies on the famous Li-Li network coding conjecture~\cite{lilinetcoding}
being true.
Despite the centrality of this conjecture, it has so forth resisted all attempts 
at either proving or refuting it.
To state the conjecture formally, we 
give a formal definition of the $k$-pairs communication problem and the Multicommodity Flow problem.
We adopt a similar exposition as in earlier works~\cite{sortinglbstoc19,lbmult}.

\paragraph{$k$-pairs Communication Problem.} To keep the definition as simple as possible, we restrict ourselves to directed acyclic communication networks/graphs and we assume that the demand between every source-sink pair is the same. This will be sufficient for our proofs. For a more general definition, we refer the reader to Adler et al.~\cite{Adler:soda}.

The input to the $k$-pairs communication problem is a directed acyclic graph $G=(V,E)$ where each edge $e \in E$ has a capacity $c(e) \in \R^+$. There are $k$ sources $s_1,\dots,s_k \in V$ and $k$ sinks $t_1,\dots,t_k \in V$. Typically there is also a demand $d_i$ between each source-sink pair, but for simplicity we assume $d_i = 1$ for all pairs. This is again sufficient for our purposes.

Each source $s_i$ receives a message $A_i$ from a predefined set of messages $A(i)$. It will be convenient to think of this message as arriving on an in-edge. Hence we add an extra node $S_i$ for each source, which has a single out-edge to $s_i$. The edge has infinite capacity.

A network coding solution specifies for each edge $e \in E$ an alphabet $\Gamma(e)$ representing the set of possible messages that can be sent along the edge. For a node $u \in V$, define $\In(u)$ as the set of in-edges at $u$. A network coding solution also specifies, for each edge $e=(u,v) \in E$, a function $f_e : \prod_{e' \in \In(u)} \Gamma(e') \to \Gamma(e)$ which determines the message to be sent along the edge $e$ as a function of all incoming messages at node $u$. Finally, a network coding solution specifies for each sink $t_i$ a decoding function $\sigma_i : \prod_{e \in \In(t_i)} \Gamma(e) \to M(i)$. The network coding solution is correct if, for all inputs $A_1,\dots,A_k \in \prod_i A(i)$, it holds that $\sigma_i$ applied to the incoming messages at $t_i$ equals $A_i$, i.e. each source must receive the intended message.

In an execution of a network coding solution, each of the extra nodes $S_i$ starts by transmitting the message $A_i$ to $s_i$ along the edge $(S_i,s_i)$. Then, whenever a node $u$ has received a message $a_e$ along all incoming edges $e=(v,u)$, it evaluates $f_{e'}(\prod_{e \in \In(u)} a_e)$ on all out-edges $e' = (u,w) \in E$ and forwards the message along the edge $e'$.

Following Adler et al.~\cite{Adler:soda} (and simplified a bit), we define the \emph{rate} of a network coding solution as follows: Let each source receive a uniform random and independently chosen message $A_i$ from $A(i)$. For each edge $e$, let $A_e$ denote the random variable giving the message sent on the edge $e$ when executing the network coding solution with the given inputs. The network coding solution achieves rate $r$ if:
\begin{itemize}
\item $H(A_i) \geq r d_i = r$ for all $i$.
\item For each edge $e \in E$, we have $H(A_e) \leq c(e)$.
\end{itemize}
Here $H(\cdot)$ denotes binary Shannon entropy. The intuition is that the rate is $r$, if the solution can handle upscaling the entropy of all messages by a factor $r$ compared to the demands.

\paragraph{Multicommodity Flow.}
A multicommodity flow problem in an undirected graph $G=(V,E)$ is specified by a set of $k$ source-sink pairs $(s_i,t_i)$ of nodes in $G$. We say that $s_i$ is the source of commodity $i$ and $t_i$ is the sink of commodity $i$. Each edge $e \in E$ has an associated capacity $c(e) \in \R^+$. In addition, there is a demand $d_i$ between every source-sink pair. For simplicity, we assume $d_i = 1$ for all $i$ as this is sufficient for our needs.

A (fractional) solution to the multicommodity flow problem specifies for each pair of nodes $(u,v)$ and commodity $i$, a flow $f_i(u,v) \in [0,1]$. Intuitively, $f_i(u,v)$ specifies how much of commodity $i$ is to be sent from $u$ to $v$. The flow satisfies \emph{flow conservation}, meaning that:
\begin{itemize}
\item For all nodes $u$ that is not a source or sink, we have $\sum_{w \in V} f_i(u,w) - \sum_{w \in V} f_i(w,u) = 0$.
\item For all sources $s_i$, we have 
$\sum_{w \in V} f_i(s_i,w) - \sum_{w \in V}f_i(w,s_i) = 1 $.
\item For all sinks we have
$\sum_{w \in V} f_i(w,t_i) - \sum_{w \in V} f_i(t_i,w) = 1 $.
\end{itemize}
The flow also satisfies that for any pair of nodes $(u,v)$ and commodity $i$, there is only flow in one direction, i.e. either $f_i(u,v)=0$ or $f_i(v,u)=0$. Furthermore, if $(u,v)$ is not an edge in $E$, then $f_i(u,v) = f_i(v,u)=0$. A solution to the multicommodity flow problem achieves a rate of $r$ if:
\begin{itemize}
\item For all edges $e=(u,v) \in E$, we have $r \cdot \sum_i d_i (f_i(u,v) + f_i(v,u)) = r \cdot \sum_i (f_i(u,v) + f_i(v,u)) \leq c(e)$.
\end{itemize}
Intuitively, the rate is $r$ if we can upscale the demands by a factor $r$ without violating the capacity constraints.

\paragraph{The Undirected $k$-pairs Conjecture.} 
The undirected $k$-pairs conjecture~\cite{lilinetcoding}
is stated below:

\begin{conj}[Undirected $k$-pairs Conjecture~\cite{lilinetcoding}]
\label{conj:main}
The coding rate is equal to the Multicommodity Flow rate in undirected graphs.
\end{conj}

This conjecture
implies the following for our setting: Given an input to the $k$-pairs communication problem, specified by a directed acyclic graph $G$ with edge capacities and a set of $k$ source-sink pairs with a demand of $1$ for every pair, let $r$ be the best achievable network coding rate for $G$. Similarly, let $G'$ denote the undirected graph resulting from making each directed edge in $G$ undirected (and keeping the capacities, source-sink pairs and a demand of $1$ between every pair). Let $r'$ be the best achievable flow rate in $G'$. Conjecture~\ref{conj:main} implies that $r \leq r'$. 

Having defined coding rate and flow rate formally, 
we also mention that the result of Braverman et al. \cite{braverman2016network} implies that if there exists a graph $G$ where the network coding rate $r$, and the flow rate $r'$ in the corresponding undirected graph $G'$, satisfies $r \geq (1+\eps)r'$ for a constant $\eps>0$, then there exists an infinite family of graphs $\{G^*\}$ for which the corresponding gap is at least $(\lg |G^*|)^c$ for a constant $c>0$. So far, all evidence suggest that no such gap exists, as formalized in Conjecture~\ref{conj:main}.

\subsection{Our Lower Bound}

\begin{theorem}[Restatement of Theorem~\ref{thm:introlb}]
Suppose that the Li-Li network coding conjecture~\cite{lilinetcoding} is true. Moreover, suppose that 
each element's payload length $w > \log_2 n - k$, and the key length $k \leq \log_2 n$.
Then, any constant fan-in, constant fan-out boolean circuit that can 
sort $n$ elements each with a $k$-bit key 
and a $w$-bit payload must have 
size at least $\Omega(nk \cdot (w - \log_2 n +k))$.
\label{thm:lb}
\end{theorem}
\begin{proof}
Consider a fixed constant fan-in, constant fan-out boolean circuit; the
topology of the circuit induces a directed graph $G$. 
Without loss of genreality, we may assume that the graph $G$ has in-degree and out-degree 2,
since any constant fan-in, constant fan-out gate can be broken up
into constant number of gates with fan-in and fan-out 2.
Given an input array $I$, 
let ${\sf keys}(I)$
denote the sequence 
of all elements' keys in the input $I$,
and ${\sf payloads}(I)$
denote the sequence of all payloads in the input $I$.
Every element in the input array $I$ 
is represented by $k + w$ input nodes in the graph $G$, 
$k$ input nodes for representing the key henceforth also called
{\it key-input-nodes},
and $w$ input nodes for representing the payload, henceforth
also called {\it payload-input-nodes}.

Without loss of generality, we may assume that $n \geq 2^k$ is a power of $2$.
We consider the set of input 
arrays where each of the $2^k$ distinct keys appears exactly $n/2^k$ times.
Therefore, for each element in the input array, its key will determine at most  
$n/2^k$ possible positions for the element to appear in the sorted output.
This means that for every payload-input-node $v$, depending
on the corresponding element's key denoted ${\sf key}$, 
the input bit assigned to $v$ must be transferred to 
an output node among $n/2^k$ possible choices --- henceforth
we use the notation
$O_{v, {\sf key}}$ to denote the set of $n/2^k$ possible choices of output nodes
for the payload-input-node $v$, determined by 
the corresponding element's key ${\sf key}$.
Let $D_{v, {\sf key}}$ denote the 
minimum distance from $v$ to 
any output node in $O_{v, {\sf key}}$. 
Observe that for any payload-input-node $v$, 
at least 
$2^k - 2^{k/2}$ choices of ${\sf key} \in \{0, 1\}^{k}$
will give  $D_{v, {\sf key}} \geq k/2$.
This is because in $k/2$ depth, $v$ can reach only $2^{k/2}$ nodes in $G$.

With this, we can conclude the following:

\begin{claim} 
Fix any subset $S$ of payload-input-nodes.
There exists
a choice of ${\sf keys}(I)$, henceforth denoted ${\sf keys}^*(I)$, 
where each of the $2^k$ distinct keys appears exactly $n/2^k$ times, 
such that the at least $1 - \frac{1}{2^{k/2}}$ fraction of the payload-input-nodes in $S$, 
denoted $v$,
must satisfy $D_{v, {\sf key}(v)} \geq k/2$
where ${\sf key}(v)$ is the key in ${\sf keys}(I)$ corresponding to $v$.
\label{clm:longpath}
\end{claim}
\begin{proof}
To see this, we can sample  ${\sf keys}(I)$
at random subject to the constraint that 
each of the $2^k$ distinct keys must appear exactly $n/2^k$ times.
For each fixed $v$, 
the probability that $D_{v, {\sf key}(v)} \geq k/2$
is at least \[\frac{2^k-2^{k/2}}{2^k} = 1 - \frac{1}{2^{k/2}}\]
Due to linearity of expectation, we conclude that 
in expectation, there is at least  
$1 - \frac{1}{2^{k/2}}$ fraction of such $v$'s in any fixed $S$ that satisfy 
$D_{v, {\sf key}(v)} \geq k/2$.  
\end{proof}

We now continue with the proof of Theorem~\ref{thm:lb}.

\paragraph{Augmenting the graph.} 
So far we have considered the directed graph $G$ that represents the circuit.
We now augment the graph $G$ into a new graph $G'$ as follows.
We add $2^k$ nodes, henceforth called {\it aggregators}.
We now add a directed edge  
from every output node
in $G$ corresponding to an element with the key $i$,
to the $i$-th aggregator.
The $i$-th aggregator node will now stably sort these 
elements (whose keys are $i$) 
based on the first $\log_2 n - k$ bits   
\elaine{this assumes w large}
of their payload, and output the sorted elements --- the corresponding output
nodes become the output nodes of 
the new graph $G'$. In total, $G'$ has 
$(k + w) n$ output nodes, corresponding to a total of $n$ elements; the output nodes
are ordered first by the element's keys, 
using the first $\log_2 n -k$ bits of the elements' payload to break ties.

Now suppose we fix the choice ${\sf key}^*(I)$ 
as mentioned above, 
for all input elements 
with the same key, we choose the first  
$\log_2 n-k$ bits of their payload 
based on the order in which they appear in the input, i.e., for 
some key $k$, the leftmost element  
with key $k$ receives $0$ as the  
first $\log_2 n-k$ bits of their payload, the second
leftmost element with key $k$ receives $1$
as the first $\log_2 n-k$ bits of their payload, and so on.
The remaining $w-\log_2  n + k$
bits of every element's payload is chosen at random.
We may consider the key bits and the first $\log_2 n-k$ bits 
of the payloads as being hard-wired into the network, and every edge has capacity 1 ---
henceforth we call this resulting network $G''$.
The resulting network solves a $k$-pairs communication problem --- which
source is paired with which destination is determined by ${\sf keys}^*(I)$
and the first $\log_2 n - k$ bits of each payload.
Moreover, the network routes $(w-\log_2 n + k) n$ input nodes 
to $(w-\log_2 n + k) n$ 
output nodes, where each input node receives a uniform random bit as input.

\paragraph{Applying the Conjecture~\ref{conj:main}.}
Now, consider the undirected version of $G''$ 
where each edge's capacity is still 1.
Based on Conjecture~\ref{conj:main}, 
the undirected version of $G''$ should solve the corresponding multi-commodity 
flow problem where each of the $(w-\log_2 n + k) n$ input nodes 
wants to route a commodity to each of the $(w-\log_2 n + k) n$ output nodes.
In the solution of the multi-commodity
flow problem, the $i$-th aggregator node in $G''$ must have 
$\frac{n}{2^k} \cdot (w-\log_2 n + k)$
amount of flow coming in; further, there is one unit of flow
corresponding to each of the last $w-\log_2 n+k$
payload bits for each element whose key is $i$.
Due to Claim~\ref{clm:longpath}, 
$1-1/2^{k/2}$
fraction of the input nodes in $G''$ 
has a path of length at least $k/2 + 1$ to 
its corresponding aggregator node.
We conclude that 
$G''$ has  
at least $\frac{n}{2^k} \cdot (w-\log_2 n + k) \cdot 2^k \cdot (k/2 + 1)$
edges not including the edges between the aggregator nodes
and the $(w-\log_2 n + k) n$ output nodes.

Thus, the graph $G''$ has at least 
$  \frac{nk}{2} \cdot (w-\log_2 n + k)$
edges excluding the aggregator nodes and all its incident edges.
This means that the original circuit's size is at least 
$\Omega(nk \cdot (w-\log_2 n + k))$.

\end{proof}

Finally, as mentioned earlier, our lower bound requires
that $w \geq \log_2 n - k$. 
Technically, this is because in our proof, we steal 
$\log_2 n -k$ bits from the payload to fix an ordering 
among the elements with the same key.
Our lower bound shows the near optimality of our construction for sufficiently large $w$.
For small $w$, it is an interesting open question whether a better upper bound exists ---
however, 
to answer this question would necessarily require us to consideration 
of algorithms that are {\it not} in the indivisible
model 
due to the $\Omega(nk)$ lower bound on the number of selector gates
for any algorithm in the indivisible model~\cite{osortsmallkey}.

%% file: conclusion.tex
\section{Conclusion and Future Work}
\label{sec:conclusion}

In this paper, we showed a theoretical generalization 
of the AKS sorting circuit. We show that for sorting
$n$ elements each described with a $k$-bit key and a $w$-bit payload, 
a circuit of size $O(n (k + w) \cdot k) \cdot \poly(\log^* n - \log^*(k+w))$
suffices. Specifically, when $k = o(\log n)$, our circuit size is asymptotically
better than AKS (ignoring $\poly\log^*$ terms).
As a special case and stepping stone to our main result,
we also show that compaction and selection can be computed with linear-sized circuits. 
We also show that our result is nearly optimal for every choice of $k$ 
as long as $k = O(\log n)$.

Our work leaves open the following future directions:
\begin{enumerate}
\item 
An obvious open question is to get rid of the extra $\poly\log^*$ terms
in the circuit size.
\item 
Another open question concerns the depth of the circuit.
In this work, we cared mostly about minimizing the circuit size, but 
not the depth.
Therefore, an open question is, {\it 
can we sort $n$ elements each described with a $k$-bit key and a $w$-bit payload,
with a circuit of size $O(nk \cdot (k+w))$
and depth $O(\log n)$?
}
We note that $\Omega(\log n)$
depth is necessary even for compaction, i.e.,  1-bit sorting. 
Observe that a compaction circuit 
can compute the logical-or of $n$ bits, which is known to have a $\Omega(\log n)$-depth
lower bound even on a Concurrent-Read-Exclusive-Write (CREW) PRAM~\cite{orlbpram}; and 
clearly, a circuit of depth $d$ can be simulated by a CREW PRAM of depth $d$. 

\item 
Our construction currently has an enormous constant. {\it Can we attain asymptotically the same
result but with smaller concrete constants?}

\item 
Finally, as mentioned earlier, another question is whether 
better upper bounds 
exist for small $w$ --- as mentioned, if such upper bounds existed, they cannot
be in the indivisible model due to 
the $\Omega(n k)$ lower bound 
on the number of selector gates in the indivisible model~\cite{osortsmallkey}.
\end{enumerate}

%% file: constants.tex

\section{Compaction Circuit for Tiny $w$}
\label{sec:tinyw}
We now describe a compaction circuit for very small $w$, e.g., when
$w$ is constant or slightly larger than a constant.
Recall that our Theorem~\ref{thm:introcompaction}
states that there is a circuit of
$O(n w) \cdot \min\left( \poly (\log^*n - \log^* w)), \ \ 2^w/w \right)$
size to compact $n$ elements each of bit-width $w$.
Section~\ref{sec:lineartc}
obtained the $O(n w) \cdot \poly (\log^*n - \log^* w)$ part of the result;
therefore this section shows
that there is a circuit of size
$O(n \cdot 2^w)$ that can compact $n$ elements each of bit-width $w$.
\elaine{changed the earlier thm to be consistent with this}

Henceforth, we assume that each distinguished element is marked with the label 0
and each non-distinguished element is labeled with 1.
Given an element of bit-width $w$, we can define an {\it extended element} whose length is  
$w+1$ bits by concatenating the element's distinguished label and its payload.
Therefore, an extended element can take value from the domain $\{0, 1, \ldots, 2\cdot 2^w -1\}$.
To achieve compaction, we will actually view all elements as extended
elements and sort all of them. 
To achieve this, we will perform the following:

\ignore{
Note that an element $v$ of bit-width $w$ can only take values from $v \in \{0, 1, \ldots, 2^w-1\}$. 
For each $v \in \{0, 1, \ldots, 2^w-1\}$,
we count the number of occurrences for $v$ that are distinguished (to be stored
in ${\sf cnt}_0[v]$),
and number of occurrences for $v$ that are non-distinguished
(to be stored in ${\sf cnt}_1[v]$).
}

\ignore{
First, count the number of occurrences 
for each of the $2\cdot 2^w$ possible extended elements; 
the result is written down in an array denoted ${\sf cnt}[0:2\cdot 2^w-1]$ which
contains $2\cdot 2^w$ entries.
We can write down all these $2 \cdot 2^w$ counters 
with a circuit of $2^w \cdot 6n$ size by Fact~\ref{fct:countprelim}. 
}

First, for each extended element $v \in \{0, 1,\ldots, 2 \cdot 2^w -1\}$, 
count how many times extended elements
of value at most $v$ have appeared in the input array. 
The result is stored in an array called ${\sf psum}$.
Now, the length of ${\sf psum}$ is $2 \cdot 2^w$
and each entry of ${\sf psum}$ is encoded with $\log n + 1$ bits.
The $v$-th entry of ${\sf psum}$ stores how many extended elements 
have value at most $v$. 
To compute ${\sf psum}$ with an $O(2^w n)$-sized circuit,
it suffices to perform the procedure below.
\begin{enumerate}
\item 
For each extended element of a value $v$,
run the binary-to-unary conversion (Fact~\ref{fct:binary_to_unary}) 
to get the $(2\cdot 2^w-1)$-dimensional vector 
such that the head $v$ coordinates are 0s and the tail $2\cdot 2^w - 1-v$ coordinates are 1s.
This conversion yields $n$ vectors (each of $2\cdot 2^w -1$ bits)
and takes $O(n \cdot 2^w)$ generalized boolean gates.

\item 
For each $i \in \{0, 1,\ldots, 2 \cdot 2^w -1\}$,
count the number of 1s in the $i$-th coordinate of all $n$ vectors,
and let the result be the $i$-th entry of ${\sf psum}$.
Using Fact~\ref{fct:countprelim},
this counting takes $O(n \cdot 2^w)$ generalized boolean gates.
\end{enumerate}

From this array ${\sf psum}$, we can generate the sorted output array 
using a binary tree with $n$ leaves. Imagine that each output position is a leaf node in the tree; 
without loss of generality, we may assume that $n$ is a power of $2$.
Initially, the root holds ${\sf psum}$ and every entry
in ${\sf psum}$ takes $\log n + 1$ bits to encode. 
The root now prepares an array ${\sf psum}_L$ 
to send to the left child by setting 
every entry in ${\sf psum}$ greater than $n/2$ to 
$n/2$. 
Therefore, every entry in ${\sf psum}_L$
takes one fewer bit to encode than the original ${\sf psum}$. 

Similarly, the root prepares an array ${\sf psum}_R$
to send to the right child
by setting every entry in ${\sf psum}$ 
smaller than or equal to $n/2$ to $0$, and 
every entry whose value is at least $n/2$ to the original value minus $n/2$.
Therefore, each entry in ${\sf psum}_R$
also takes one fewer bit to encode than the original ${\sf psum}$.

We continue this process at every level of the tree (where the root is assumed
to be at level 0).
Each node in level $i$ sends an array to its left child and right child, 
and the number of bits needed to encode each entry  
of the array is $\log n - i$. 
An array ${\sf psum}'$ 
received by a node in the tree always encodes the prefix sums within its own subtree;
specifically ${\sf psum}'[v]$ encodes how many times extended elements at most $v$
have appeared in its subtree.

Finally, when a leaf node receives an incoming array ${\sf psum}'$, scan through the array to 
find the first index $v$  
such that  ${\sf psum}'[v]$ is non-zero 
and output $v$. 
The final sorted array 
is obtained by concatenating all leaves' output values from left to right.

The computation done by each node at level $i$ in the tree
can be implemented with an $O(2^w (\log n - i))$-sized circuit.

Summarizing the above, 
it is not hard to see that the entire computation  
can be implemented 
with an $O(2^w n)$-sized circuit.


\section{Optimizations and Tightened Theorem Statement}
\label{sec:tighten}

In this section, 
we fine-tune the parameters in the construction 
of Section~\ref{sec:tightfromloose},~\ref{sec:looseFromTight}, and~\ref{sec:lineartc}
to minimize the polynomial $\poly(\cdot)$
in Theorem~\ref{thm:main}.
As stated in the following,
this $\poly(\cdot)$ can be chosen to nearly quadratic.

\begin{theorem}[Linear-sized tight compaction, tightened]
For any constant $\gamma > 0$,
there exists a constant fan-in, constant fan-out boolean circuit 
that solves $(n, w)$-tight compaction and the total
number of boolean gates is upper bounded by $\max\left(1,(\log^* n - \log^* w)^{2+\gamma}\right) \cdot O(n \cdot w)$.
\label{thm:optimized_result}
\end{theorem}

\begin{proof}
The case $w > \log n$ is 
easy (see Footnote~\ref{fnt:bigw})
so henceforth we focus on the case $w \le \log n$.
We begin with re-parameterizing the construction with some $\epsilon_1, \epsilon_2, \ldots$,
analyzing the constraints between the parameters,
and then choosing them appropriately to satisfy the given $\gamma > 0$.

The construction is re-parameterized as follows with 
our previous parameters noted.
\begin{itemize}
\item 
In Step~\ref{step:swap} of Algorithm~\ref{alg:swap}, 
call an $(n,w)$-loose swapper such that at most $\epsilon_1$ fraction
of resulting array remain colored.
(Previously $\epsilon_1 := 1/128$.)

\item 
In Section~\ref{sec:loosecompact}, we aim to construct an $(n,w)$-loose compactor
such that takes as input at most $\epsilon_1 \cdot n$ real elements
and outputs an array of $\epsilon_2 \cdot n$ elements.
We re-parameterize the construction in Section~\ref{sec:loosecompact}
as follows.
(Previously $\epsilon_2 := 1/2$.)
  \begin{itemize}
  \item 
  In Step~\ref{step:countperchunk} of Section~\ref{sec:loosecompact}, 
  define a chunk as \emph{sparse} if there are at most $\epsilon_3 \cdot f(n)$
  real elements in it; otherwise define it \emph{dense}.
  (Previously $\epsilon_3 := 1/32$.)

  \item 
  Hence, at most $\epsilon_4 \cdot \frac{n}{f(n)}$ chunks are dense,
  and at least $(1- \epsilon_4)$ fraction of chucks are sparse,
  where $\epsilon_4 := \epsilon_1 / \epsilon_3$.
  Use $\epsilon_4$ and $1-\epsilon_4$ as thresholds in 
  Step~\ref{step:compactperchunk} of Section~\ref{sec:loosecompact}.
  (Previously $\epsilon_4 := 1/4$.)
  \end{itemize}
\end{itemize}

We then analyze the constraints on the above parameters.
As noted in Theorem~\ref{thm:lswap}, for any constant $\epsilon_1 > 0$,
there is a linear sized $(n,w)$-loose swapper,
so $\epsilon_1$ is not constrained.
Then, by the above parameters $\epsilon_3$ and $\epsilon_4$,
the output size of Section~\ref{sec:loosecompact} is
$\epsilon_4 \cdot n + \epsilon_3 f(n) \cdot \ceil{(1- \epsilon_4)\cdot \frac{n}{f(n)}}$,
which aimed to be at most $\epsilon_2 \cdot n$.
For any constant $\epsilon_4 \in (0,1)$, $f(x) \leq \log x$,
there exists a constant $n_0 \in \N$ such that 
$\ceil{(1- \epsilon_4)\cdot \frac{n}{f(n)}} \leq \frac{n}{f(n)}$
for all $n > n_0$.
Hence, we need $\epsilon_3 + \epsilon_4 \leq \epsilon_2$
to achieve the $\epsilon_2 \cdot n$ output size.

Next, we inductively calculate the circuit size in the bootstrapping 
of Section~\ref{sec:repeated_bootstrap}.
The bootstrapping starts from Algorithm~\ref{alg:loose-compaction},
which takes $C(1) \cdot n \log_2 n$ generalized boolean gates
and $C(1) \cdot n$ number of $w$-selector gates for some constant $C(1)$.
Let $c_1 \cdot n$ be the number of generalized boolean gates
and $c_2 \cdot n$ be the number of $w$-selector gates
of the $(n,w)$-loose swapper.
After $d$ times of repeated bootstrapping,
assume for induction that the loose compactor takes
$C(d)\cdot n f(n)$ generalized boolean gates
and $C(d)\cdot n$ number of $w$-selector gates,
where $C(d)$ is a function of $d$ (rather than $n$).
To construct a tight compactor from such loose compactor,
we count and color elements, which takes $c_3 n$ boolean gates 
for some constant $c_3$ in Section~\ref{sec:tightcompact};
Then, we run the loose swapper and the $(n,w+1)$-loose compactor 
recursively, where problem size $n$ is decreasing by $\epsilon_2$.
Hence, the tight compactor takes following costs.
\begin{align*}
\mbox{generalized boolean gates: } &
\frac{1}{1- \epsilon_2} \cdot \big(2c_1n + C(d)\cdot n f(n) + (c_2+c_3+C(d))\cdot n\big)
\\
\mbox{$w$-selector gates: } &
\frac{1}{1- \epsilon_2} \cdot \big(c_2n + 2C(d)\cdot n\big)
\end{align*}
We continue to construct the $(d+1)$-bootstrapped loose compactor
by marking dense and sparse chunks,
which takes $c_4n$ boolean gates.
Afterwards, we run 1 instance of $(\ceil{n/f(n)}, f(n)\cdot w)$-tight compactor
and $\ceil{(1- \epsilon_4)\cdot \frac{n}{f(n)}}$ instances of 
$(f(n), w)$-tight compactor.
Let $\tilde n_1 := \ceil{n/f(n)}$ and $\tilde n_2 := \ceil{(1- \epsilon_4)\cdot \frac{n}{f(n)}}$ for short.
Then, the number of generalized boolean gates is
\begin{align*}
c_4n
& + \frac{1}{1- \epsilon_2} \cdot 
  \big(2c_1 \tilde n_1 + C(d)\cdot \tilde n_1 f(\tilde n_1) + (c_2+c_3+C(d))\cdot \tilde n_1\big) 
\\
& + \frac{1}{1- \epsilon_2} \cdot 
  \big(2c_1 f(n) + C(d)\cdot f(n) f(f(n)) + (c_2+c_3+C(d))\cdot f(n)\big) \cdot \tilde n_2
\\
& \leq \frac{1}{1- \epsilon_2} \cdot 
  \big(2C(d)\cdot n + C(d) \cdot f(f(n)) \cdot n + C(d) \cdot \tilde n_1 + c_5 n\big),
\end{align*}
where $c_5 = c_4 + 4c_1 + 2c_2+2c_3$ is a constant,
and the rounding up in $\tilde n_1$ and $\tilde n_2$ are absorbed by $\epsilon_4$ for all $n > n_0$.
Similarly, the number of $w$-selector gates is
\begin{align*}
\frac{1}{1- \epsilon_2} \cdot \big(c_2 \tilde n_1 + 2C(d)\cdot \tilde n_1\big) \cdot f(n)
+ \frac{1}{1- \epsilon_2} \cdot \big(c_2 f(n) + 2C(d)\cdot f(n)\big) \cdot \tilde n_2
\\
\leq \frac{1}{1- \epsilon_2} \cdot \big(4C(d)\cdot n + 2c_2 n\big).
\end{align*}
We choose $C(d+1) := C(d) \cdot \frac{4}{1- \epsilon_2} + \frac{c_5+2c_2}{1 -\epsilon_2}$
so that the number of generalized boolean gates is at most $C(d+1) \cdot f(f(n)) \cdot n$
and the number of $w$-selector gates is at most $C(d+1) \cdot n$.
As the base case $C(1)$ is a constant,
the recursion of $C(d)$ solves to $C(d) = O\left((\frac{4}{1 -\epsilon_2})^d\right)$.



Putting together, after repeated bootstrapping for $d$ times,
the multiplicative overhead is $\left(\frac{4}{1- \epsilon_2}\right)^d$,
which equals to $(\log^*n - \log^* w)^{2+\log_2 \frac{1}{1- \epsilon_2}}$
as $d = \log_2 (\log^*n - \log^* w)$.
Hence, for any $\gamma > 0$,
it suffices to choose $\epsilon_2 = \frac{2^\gamma - 1}{2^\gamma}$,
$\epsilon_3 = \epsilon_4 = \epsilon_2 / 2$,
and $\epsilon_1 = \epsilon_3 \cdot \epsilon_4$
to satisfy all the above constraints.

We remark that in Algorithm~\ref{alg:loose-compaction}, 
we need that the input to loose compaction is 
sparse enough (i.e., 1/128) so the {\sf ProposeAcceptFinalize} works.
However, any $\epsilon_1 \leq \epsilon_2$ works
in Section~\ref{sec:loosecompact} because we use tight compactor directly.
Hence, Algorithm~\ref{alg:loose-compaction} is only used in 
the base case of the repeated bootstrapping 
(and we do not tune its parameters).
\end{proof}

%% file: main-smallkey-sort.bbl
\newcommand{\etalchar}[1]{$^{#1}$}
\begin{thebibliography}{AKL{\etalchar{+}}20b}

\bibitem[AFKL19]{lbmult}
Peyman Afshani, Casper~Benjamin Freksen, Lior Kamma, and Kasper~Green Larsen.
\newblock Lower bounds for multiplication via network coding.
\newblock In {\em 46th International Colloquium on Automata, Languages, and
  Programming, {ICALP} 2019, July 9-12, 2019, Patras, Greece.}, pages
  10:1--10:12, 2019.

\bibitem[AHJ{\etalchar{+}}06]{Adler:soda}
Micah Adler, Nicholas J.~A. Harvey, Kamal Jain, Robert Kleinberg, and
  April~Rasala Lehman.
\newblock On the capacity of information networks.
\newblock In {\em Proceedings of the Seventeenth Annual ACM-SIAM Symposium on
  Discrete Algorithm}, SODA '06, pages 241--250, 2006.

\bibitem[AHNR98]{sortlinearandersson}
Arne Andersson, Torben Hagerup, Stefan Nilsson, and Rajeev Raman.
\newblock Sorting in linear time?
\newblock {\em J. Comput. Syst. Sci.}, 57(1):74--93, August 1998.

\bibitem[AKL{\etalchar{+}}20a]{optorama}
Gilad Asharov, Ilan Komargodski, Wei-Kai Lin, Kartik Nayak, Enoch Peserico, and
  Elaine Shi.
\newblock {OptORAMa}: Optimal {Oblivious} {RAM}.
\newblock In {\em Advances in Cryptology - {EUROCRYPT} 2020}, 2020.
\newblock To appear. See also: \url{https://eprint.iacr.org/2018/892}.

\bibitem[AKL{\etalchar{+}}20b]{paracompact}
Gilad Asharov, Ilan Komargodski, Wei-Kai Lin, Enoch Peserico, and Elaine Shi.
\newblock Oblivious parallel tight compaction.
\newblock In {\em Information-Theoretic Cryptography (ITC)}, 2020.

\bibitem[AKS83]{aks}
M.~Ajtai, J.~Koml\'{o}s, and E.~Szemer{\'e}di.
\newblock An {O}(n log n) sorting network.
\newblock In {\em STOC}, 1983.

\bibitem[Ale69]{sortminmem}
V.E. Alekseev.
\newblock Sorting algorithms with minimum memory.
\newblock {\em Kibernetica}, 5:99--103, 1969.

\bibitem[ALM90]{alm90}
Sanjeev Arora, Frank~Thomson Leighton, and Bruce~M. Maggs.
\newblock On-line algorithms for path selection in a nonblocking network
  (extended abstract).
\newblock In {\em Proceedings of the 22nd Annual {ACM} Symposium on Theory of
  Computing, May 13-17, 1990, Baltimore, Maryland, {USA}}, 1990.

\bibitem[Bat68]{bitonicsort}
Kenneth~E. Batcher.
\newblock Sorting networks and their applications.
\newblock In {\em American Federation of Information Processing Societies:
  {AFIPS} Conference Proceedings: 1968 Spring Joint Computer Conference,
  Atlantic City, NJ, USA, 30 April - 2 May 1968}, pages 307--314, 1968.

\bibitem[BFP{\etalchar{+}}72]{blummedian}
Manuel Blum, Robert~W. Floyd, Vaughan Pratt, Ronald~L. Rivest, and Robert~E.
  Tarjan.
\newblock Linear time bounds for median computations.
\newblock In {\em Proceedings of the Fourth Annual ACM Symposium on Theory of
  Computing}, STOC '72, pages 119--124, New York, NY, USA, 1972. ACM.

\bibitem[BFP{\etalchar{+}}73]{blumselect}
Manuel Blum, Robert~W. Floyd, Vaughan Pratt, Ronald~L. Rivest, and Robert~E.
  Tarjan.
\newblock Time bounds for selection.
\newblock {\em J. Comput. Syst. Sci.}, 7(4):448--461, August 1973.

\bibitem[BGS17]{braverman2016network}
Mark Braverman, Sumegha Garg, and Ariel Schvartzman.
\newblock Coding in undirected graphs is either very helpful or not helpful at
  all.
\newblock In {\em 8th Innovations in Theoretical Computer Science Conference,
  {ITCS} 2017, January 9-11, 2017, Berkeley, CA, {USA}}, pages 18:1--18:18,
  2017.

\bibitem[BH91]{parallelhash}
Holger Bast and Torben Hagerup.
\newblock Fast and reliable parallel hashing.
\newblock In {\em Proceedings of the Third Annual ACM Symposium on Parallel
  Algorithms and Architectures}, SPAA '91, pages 50--61, 1991.

\bibitem[BH93]{parallelintegersorting}
H.~Bast and Torben Hagerup.
\newblock Fast parallel space allocation, estimation and integer sorting.
\newblock {\em Inf. Comput.}, 123:72--110, 1993.

\bibitem[bmm]{bmm}
Private communication with Bruce Maggs.

\bibitem[BN16]{isthereoramlb}
Elette Boyle and Moni Naor.
\newblock Is there an oblivious {RAM} lower bound?
\newblock In {\em ITCS}, 2016.

\bibitem[CDR86]{orlbpram}
Stephen~A. Cook, Cynthia Dwork, and R{\"{u}}diger Reischuk.
\newblock Upper and lower time bounds for parallel random access machines
  without simultaneous writes.
\newblock {\em {SIAM} J. Comput.}, 15(1):87--97, 1986.

\bibitem[DO20]{DittmerOstrovsky20}
Samuel Dittmer and Rafail Ostrovsky.
\newblock Oblivious tight compaction in {O}(n) time with smaller constant.
\newblock Cryptology ePrint Archive, Report 2020/377, 2020.
\newblock \url{https://eprint.iacr.org/2020/377}.

\bibitem[FHLS19]{sortinglbstoc19}
Alireza Farhadi, MohammadTaghi Hajiaghayi, Kasper~Green Larsen, and Elaine Shi.
\newblock Lower bounds for external memory integer sorting via network coding.
\newblock In {\em STOC}, 2019.

\bibitem[GK96]{ppmsort}
Michael~T. Goodrich and S.~Rao Kosaraju.
\newblock Sorting on a parallel pointer machine with applications to set
  expression evaluation.
\newblock {\em J. ACM}, 43(2):331–361, March 1996.

\bibitem[Goo14]{zigzag}
Michael~T. Goodrich.
\newblock Zig-zag sort: A simple deterministic data-oblivious sorting algorithm
  running in {O(N Log N)} time.
\newblock In {\em STOC}, 2014.

\bibitem[Han04]{hansort01}
Yijie Han.
\newblock Deterministic sorting in o(\emph{n}loglog\emph{n}) time and linear
  space.
\newblock {\em J. Algorithms}, 50(1):96--105, 2004.

\bibitem[Han07]{hanselection}
Yijie Han.
\newblock Optimal parallel selection.
\newblock {\em {ACM} Trans. Algorithms}, 3(4):38, 2007.

\bibitem[HT02]{hansort00}
Yijie Han and Mikkel Thorup.
\newblock Integer sorting in 0(n sqrt (log log n)) expected time and linear
  space.
\newblock In {\em FOCS}, 2002.

\bibitem[JM92]{jimboselect}
Shuji Jimbo and Akira Maruoka.
\newblock Selection networks with $8n$ $\log_2 n$ size and $o(\log n)$ depth.
\newblock In {\em Algorithms and Computation}, pages 165--174, 1992.

\bibitem[Knu73]{knuthbook}
Donald~E. Knuth.
\newblock {\em The Art of Computer Programming, Volume {III:} Sorting and
  Searching}.
\newblock Addison-Wesley, 1973.

\bibitem[KR81]{Kirkpatricksort}
David~G. Kirkpatrick and Stefan Reisch.
\newblock Upper bounds for sorting integers on random access machines.
\newblock Technical report, 1981.
\newblock University of British Columbia.

\bibitem[LL04]{lilinetcoding}
Zongpeng Li and Baochun Li.
\newblock Network coding : The case of multiple unicast sessions.
\newblock In {\em Allerton Conference on Communications}, volume~16, page~8,
  2004.

\bibitem[LMS95]{leightonselection}
Tom Leighton, Yuan Ma, and Torsten Suel.
\newblock On probabilistic networks for selection, merging, and sorting.
\newblock In {\em Proceedings of the Seventh Annual ACM Symposium on Parallel
  Algorithms and Architectures}, SPAA '95, pages 106--118, 1995.

\bibitem[LSX19]{osortsmallkey}
Wei{-}Kai Lin, Elaine Shi, and Tiancheng Xie.
\newblock Can we overcome the $n \log n$ barrier for oblivious sorting?
\newblock In {\em SODA}, 2019.

\bibitem[MZ14]{odsmitchell}
John~C. Mitchell and Joe Zimmerman.
\newblock {D}ata-{O}blivious {D}ata {S}tructures.
\newblock In {\em STACS}, pages 554--565, 2014.

\bibitem[Pat90]{akspaterson}
M.~S. Paterson.
\newblock Improved sorting networks with $o(\log n)$ depth.
\newblock In {\em Algorithmica}, 1990.

\bibitem[Pip90]{pippengerselect}
Nicholas Pippenger.
\newblock Selection networks.
\newblock In {\em Algorithms}, pages 2--11, Berlin, Heidelberg, 1990. Springer
  Berlin Heidelberg.

\bibitem[Pip96]{selfroutingsuperconcentrator}
Nicholas Pippenger.
\newblock Self-routing superconcentrators.
\newblock {\em J. Comput. Syst. Sci.}, 52(1):53--60, February 1996.

\bibitem[PPRY18]{panorama}
Sarvar Patel, Giuseppe Persiano, Mariana Raykova, and Kevin Yeo.
\newblock Panorama: Oblivious ram with logarithmic overhead.
\newblock In {\em FOCS}, 2018.

\bibitem[RWZ88]{ssa}
B.~K. Rosen, M.~N. Wegman, and F.~K. Zadeck.
\newblock Global value numbers and redundant computations.
\newblock In {\em Proceedings of the 15th ACM SIGPLAN-SIGACT Symposium on
  Principles of Programming Languages}, POPL '88, pages 12--27, New York, NY,
  USA, 1988. ACM.

\bibitem[Sav97]{savagebook}
John~E. Savage.
\newblock {\em Models of Computation: Exploring the Power of Computing}.
\newblock Addison-Wesley Longman Publishing Co., Inc., Boston, MA, USA, 1st
  edition, 1997.

\bibitem[Sei09]{aksjoel}
Joel Seiferas.
\newblock Sorting networks of logarithmic depth, further simplified.
\newblock {\em Algorithmica}, 53(3):374–384, March 2009.

\bibitem[Tho02]{thorupsort}
Mikkel Thorup.
\newblock Randomized sorting in {$O(n \log \log n)$} time and linear space
  using addition, shift, and bit-wise boolean operations.
\newblock {\em J. Algorithms}, 42(2):205--230, 2002.

\bibitem[Val75]{comparetree}
Leslie~G. Valiant.
\newblock Parallelism in comparison problems.
\newblock {\em {SIAM} J. Comput.}, 4(3):348--355, 1975.

\bibitem[WZ99]{beateigen}
Avi Wigderson and David Zuckerman.
\newblock Expanders that beat the eigenvalue bound: explicit construction and
  applications.
\newblock {\em Combinatorica}, 19(1):125--138, 1999.

\bibitem[Yao80]{yaoselectnet}
Andrew~Chi{-}Chih Yao.
\newblock Bounds on selection networks.
\newblock {\em {SIAM} J. Comput.}, 9(3):566--582, 1980.

\end{thebibliography}
